\documentclass[11pt]{article}
\usepackage[margin = 1in]{geometry}

\newcount\Comments  
\newcommand{\kibitz}[2]{\ifnum\Comments=1{\color{#1}{#2}}\fi}
\newcommand{\hma}[1]{\kibitz{blue}{[HONGYAO: #1]}}
\newcommand{\rmr}[1]{\kibitz{red}{[RESHEF: #1]}}

\newcommand{\dcp}[1]{\kibitz{cyan}{[DAVID: #1]}}

\newcount\vsqCounter
\newcommand{\vsq}[1]{\ifnum\vsqCounter=1{\vspace{#1}}\fi}

\usepackage[numbers]{natbib}

\usepackage{etex}
\usepackage{latexsym}
\usepackage{amssymb}
\usepackage{amsmath}
\usepackage{amsfonts}
\usepackage{bbm}
\usepackage{ifthen}
\usepackage{enumerate}
\usepackage{makeidx}
\usepackage{psfrag}
\usepackage{floatflt}
\usepackage{verbatim}
\usepackage{dsfont}
\usepackage{scalefnt}
\usepackage{epsfig}
\usepackage{epstopdf}
\usepackage{color} 
\usepackage{xcolor}
\usepackage{calc}
\usepackage{graphicx} 
\usepackage{subfig}
\usepackage{multicol}
\usepackage{tikz}
\usetikzlibrary{patterns, calc, intersections, through, backgrounds}
\usetikzlibrary{decorations.pathreplacing}

\usepackage{amsthm}		
\usepackage{thmtools, thm-restate}

\usepackage{algorithmic}
\usepackage[linesnumbered,ruled,vlined]{algorithm2e}

\definecolor{auburn}{rgb}{0.43, 0.21, 0.1}
\colorlet{darkblue}{blue!35!black}
\usepackage[colorlinks,
	citecolor = darkblue,
	urlcolor = darkblue,
	linkcolor = darkblue]{hyperref}

\theoremstyle{plain}
\newtheorem{theorem}{Theorem}

\theoremstyle{definition}
\newtheorem{definition}{Definition} 
\newtheorem{example}{Example}
\newtheorem{lemma}{Lemma}
\newtheorem{proposition}{Proposition}
\newtheorem{corollary}{Corollary}

\newcommand{\items}{M}				

\newcommand{\CDF}{F} 						

\newcommand{\fixedV}{w}				
\newcommand{\fixedP}{p}				

\newcommand{\zc}{z^0}		
\newcommand{\yzc}{y^0}		

\newcommand{\maxZ}{Z}						
\newcommand{\socialV}{W}					

\newcommand{\0}{^{(0)}}						
\newcommand{\1}{^{(1)}}						

\newcommand{\tzero}{t\0}					
\newcommand{\tone}{t\1}						

\newcommand{\mech}{\mathcal{M}}		
\newcommand{\winner}{{i^\ast}} 		
\newcommand{\wone}{{i^\ast_1}}
\newcommand{\wtwo}{{i^\ast_2}}

\newcommand{\report}{r}					
\newcommand{\reportSet}{\mathcal{R}}	

\newcommand{\bzero}{b\0}				
\newcommand{\bone}{b\1}				

\newcommand{\ut}{ut}					
\newcommand{\sw}{sw}					
\newcommand{\rev}{\mathit{rev}}		

\newcommand{\paymentSpace}{\mathcal{P}}

\newcommand{\largeCnst}{X}

\newcommand{\valeps}{V^{(\eps)}}
\newcommand{\CDFeps}{\CDF^{(\eps)}}
\newcommand{\ueps}{u^{(\eps)}}
\newcommand{\infeps}{\tau^{(\eps)}}

\newcommand{\alloc}{q}

\newcommand{\alloceps}{\alloc'}
\newcommand{\zeps}{z'}
\newcommand{\yeps}{y'}

\newcommand{\xpt}{\hat{z}}

\newcommand{\setR}{\mathbb{R}}

\newcommand{\txtwp}{~\mathrm{w.p.}~}
\newcommand{\txtif}{~\mathrm{if}~}
\newcommand{\txtfor}{~\mathrm{for}~}

\newcommand{\txtand}{~\mathrm{and}~}

\renewcommand{\th}{^{\mathrm{th}}}

\newcommand{\one}[1]{\mathds{1} \{ #1\}}

\newcommand{\E}[1]{\mathbb{E}\left[ #1 \right]}
\newcommand{\Pm}[1]{\mathbb{P}\left[ #1 \right]}

\newcommand{\txtSP}{{\mathrm{SP}}}

\newcommand{\txtCSP}{{\mathrm{CSP}}}
\newcommand{\txtCPM}{\mathrm{CP}}

\newcommand{\txtCPZ}{\mathrm{CP}(\maxZ)}

\newcommand{\eps}{\varepsilon}

\providecommand{\pwfun}[1]{\left\lbrace \begin{array}{ll} #1 \end{array} \right.}

\title{Contingent Payment Mechanisms for Resource Utilization%
\thanks{The authors thank Nick Arnosti, Yakov Babichenko, Ido Erev, Thibaut Horel, Scott Kominers, Debmalya Mandal, Jake Marcinek, Di Pei, Ilya Segal, Moshe Tennenholtz, Rakesh Vohra, and seminar and workshop participants, 
for helpful comments and discussions.}
}

\author{Hongyao Ma%
\thanks{John A. Paulson School of Engineering and Applied Sciences, Harvard University, Cambridge,
MA, 02138, USA. Email: \{hma@seas, parkes@eecs\}.harvard.edu.} 
\and 
Reshef Meir%
\thanks{Department of Industrial Engineering and Management,
Technion - Israel Institute of Technology, Technion City, Haifa 3200003, Israel. Email: reshefm@ie.technion.ac.il.}
\and
David C. Parkes$^\dagger$
\and
James Zou%
\thanks{Department of Biomedical Data Science, Stanford University, Stanford, CA 94305. Email: jamesyzou@gmail.com.}
}

\begin{document}

\maketitle

\begin{abstract}
We introduce the problem of assigning resources to improve their utilization. The motivation comes from settings where agents have uncertainty about their own values for using a resource,  and where it is in the interest of a group  that resources be used and not wasted.
Done in the right way, improved utilization maximizes social welfare--- balancing the utility of a high value but unreliable agent with the group's preference that resources be used.
%
We introduce the family of {\em contingent payment mechanisms} (CP), which may charge an agent contingent on use (a penalty). A CP mechanism is parameterized by a maximum penalty, and has a  dominant-strategy equilibrium. 
Under a set of axiomatic properties, we establish welfare-optimality for the special case $\mathrm{CP}(\socialV)$, with CP instantiated for a maximum penalty equal to societal value $\socialV$ for utilization. $\mathrm{CP}(\socialV)$ is not dominated for expected welfare by any other mechanism, and second, amongst mechanisms that always allocate the resource and have a simple indirect structure, $\mathrm{CP}(\socialV)$ strictly dominates every other mechanism. 
The special case with no upper bound on penalty, the {\em contingent second-price mechanism}, 
maximizes utilization. We extend the mechanisms to assign multiple, heterogeneous resources, and present a simulation study of the welfare properties of these mechanisms.

\end{abstract}


\section{Introduction} \label{sec:intro}

Allocated resources often go to waste, even when in scarce supply. It
is common in university departments, for example, to find that all
rooms are fully booked in advance, yet walking down the corridor one
sees that many rooms are in fact empty.
For another university related example, one of the authors of the
present paper received the following email: \vspace{-0.2em}

{\small
\begin{quote}
{\tt

SITE VISIT: XXXXX Corporation

Date:  Wednesday, January 17, 2018, 9:00am to 1:00pm

Location: XXXXXXX, XX

RESERVATIONS: Reservations are now open. Reserve your spot today! 

COST: \$15 fee to hold your reservation. There is no charge for the site-visit. You will only be charged if you cancel within a week before the trip or do not show up on the morning of the visit.


NUMBER OF PARTICIPANTS: 25 spots

%
}
\end{quote}
}

Another example that we know about considers very costly, biolab
equipment (costing as much as \$3M/yr to run), that is shared amongst
four groups and that many students use. A problem with the current,
first-come first-served reservation system is that the equipment is
always fully booked, but frequently not used when the lab technician checks.
For examples from other domains, consider allocating spots in a
spinning class to gym members, and assigning time slots for a public
electric vehicle charging station to residents in a neighborhood. Even a gym member who is highly uncertain about being able or willing to attend the class, or a resident unsure about actually needing to use the charging station, may reserve a space just in case this turns out to be convenient.


What is common to these problems is the presence of uncertainty, self-interest and down-stream utilization decisions on the part of participants, together with the broader interest of a group (a department, the corporation, or the citizens of a city) that a resource be used and not wasted: 
%
%
%
utilization often has positive externalities beyond the immediate agents, e.g. there will be less air pollution when electrical cars are charged and used; the planner might also be interested only in the utilization of the resources, e.g. the firm benefiting from potentially attracting and hiring more students if more showed up for the site visit.

We formalize the desire for utilization by introducing an additional welfare gain of $\socialV \geq 0$ when a resource is used by the assigned agent, and adopt the design objective of maximizing expected total welfare.
%
The societal value $\socialV$ models either the positive externality on the society from utilization, or the weight assigned to the planner while trading-off agents' vs. planner's welfare. 
In the special case where $\socialV = \infty$, the goal becomes one of optimizing for the planner and maximizing utilization--- the probability that the resource is used.

Despite appearing important to practice and simple to state, this problem does not appear to have been formally defined or studied in the literature.
Collecting bids and running a second-price (SP) auction
need not assign a resource to a reliable agent (the agent with the
highest expected value for the option of using a resource need not be the one most likely to use the resource).
Moreover, the SP auction does not charge payments contingent on
whether or not the resource is used, and because of this misses the
opportunity to ``shape'' incentives to use a resource once it has been assigned. A penalty of \$15 for not using a resource changes the calculus for an assigned agent: now a rational agent will choose to use the resource as long as her realized value is greater than -\$15.

Beyond our opening examples, penalties for not using a resource are
used by some hospitals in charging patients for missing an
appointment,%
\footnote{\url{https://huhs.harvard.edu/sites/default/files/HDS\%20New\%20Patient\%20Welcome\%20Letter.pdf}, visited May 10th, 2018.} 
by organizers of some conferences that charge a deposit which is returned only to students that actually attended talks,%
\footnote{\url{https://risingstarsasia2018.ust.hk/guidelines.php}, visited May 10th, 2018.} 
and by some restaurants who charge a fee if guests reserve but do not show up.%
\footnote{\url{https://www.theguardian.com/lifeandstyle/2018/feb/25/restaurants-turn-up-heat-on-no-show-diners}, visited May 10th, 2018.}
These approaches can be viewed as simple, first-come first-served schemes, and where it is not clear how the penalty should be set: a penalty that is too small is not effective, whereas a penalty that is too big will drive away participation in the scheme. 
We are not aware of any formal analysis of these kinds of mechanisms, or a design approach that takes into account the maximum penalty that an individual participant would be willing to face, which in fact is a very good signal for her reliability. This is the
main conceptual contribution of our paper.

\subsection{Our Results}

We formalize the problem of designing mechanisms for improving
resource utilization, and define a family of two-period mechanisms
that make use of payments that are contingent on whether or not a
resource is used. In our model, an agent's private type corresponds
to a distribution on her future value for using the resource--- this value models her utility from using the resource minus the utility from her outside option, and as a result may be negative. 
In period zero, agents make reports that communicate information about their type. A mechanism assigns the resource, and may both collect a payment at this time as well as determine a penalty
for the assigned agent in the event the resource is not used.
In period one, the assigned agent's value is realized, and with
knowledge of the penalty the agent decides whether or not to use the
resource.

We model the societal value for the resource being utilized as
$\socialV \geq 0$, and take as the design objective that of maximizing expected social welfare: the sum of the expected value to the assigned agent and the expected value to society. 
In the special case that $\socialV = \infty$ the design objective is to maximize utilization.
We insist on voluntary participation, and also the mechanism being
no-deficit, thus precluding charging very large penalties while also
paying the agents a very large reward to participate in the first
place.\footnote{Without the requirement of no-deficit, a simple second price auction for ``the option to use a resource and also get paid $\socialV$ when the resource is used" is welfare-optimal.}

We introduce the class of {\em contingent payment mechanisms} (CP),
parameterized by a maximum penalty $Z$. The CP mechanism has a simple, dominant-strategy equilibrium, where each agent either bids on a base payment she always pays, if she is willing to accept $Z$ as the no-show penalty, or otherwise bids on the maximum penalty she is willing to accept (Theorem~\ref{thm:dominant_strategy}). 
The main results establish the welfare-optimality of the CP mechanism
when instantiated for a maximum penalty equal to the societal value
$\socialV$ for utilization, and under a set of axioms.
First, we show that $\mathrm{CP}(\socialV)$ is not dominated
for expected welfare by any other mechanism (Theorem~\ref{thm:cpm_not_dom}). Second, we show that amongst mechanisms that always allocate the resource and support a simple indirect structure, $\mathrm{CP}(\socialV)$ optimizes social welfare profile by profile (Theorem~\ref{thm:cpm_opt}).
We formalize the simple indirect structure as requiring that mechanisms have an ordered payment space, so that all agents agree on which of any two pairs of (upfront, penalty) payments is more preferable. Such mechanisms simply ask each agent for the maximum such (upfront, penalty) payment that she is willing to accept. As a special case (Theorem~\ref{thm:CSP_beat_SP}), the societal welfare of $\mathrm{CP}(\socialV)$ dominates the SP auction.
We need a genericity assumption to state the main results, precluding ties under the mechanism. We
can dispense with this requirement in a direct-revelation analogue of
the CP mechanism.


As an interesting, and we think practically-motivated special case, the {\em contingent second-price}  (CSP)  mechanism (where the resource is assigned to the agent with the maximal willingness to pay a penalty, and the penalty faced by that agent is the second-highest such bid)  is the special case of the CP mechanism with no upper bound on penalty. 
Based on this, we obtain as an immediate corollary of Theorem~\ref{thm:cpm_opt} that CSP's utilization strictly dominates that of all other mechanisms (under the same assumptions).  The CSP mechanism also has the  appealing property that it never
collects a payment from an agent who uses the resource
(Theorem~\ref{thm:csp_uniq_opt} gives a uniqueness result for CSP
under this additional no-charge requirement).


We extend the mechanisms to the setting of multiple, heterogeneous resources (where each agent gets at most one resource) in Section~\ref{sec:multi_resource}, and present simulation results in Section~\ref{sec:simulations} to demonstrate the effectiveness of our mechanisms, comparing with second-price auctions and other benchmarks.\footnote{For assigning multiple heterogeneous resources, the generalized CP mechanisms  are dominant-strategy incentive-compatible, however, the optimality results do not generalize, and we can construct examples to show that the VCG mechanism can achieve better expected welfare or utilization. 
Still, simulation results demonstrate significantly better performance on average.}
In particular, we show that a significant improvement in societal welfare can be achieved by the CP mechanism (c.f., improvement in utilization for the CSP mechanism).

\subsection{Related Work} \label{subsec:related_work}

%
Contingent payments have arisen in previous work on auction design. Prominent examples include auctioning oil drilling licenses~\cite{hendricks1988empirical}, royalties~\cite{caves2003contracts,deb2014implementation}, ad auctions~\cite{varian2007position}, and selling a firm~\cite{ekmekci2016just}.
Unlike in our model, payments are contingent on some observable world state (e.g. amount of oil produced, a click, or the ex post cash flow) 
 rather than an agent's own downstream actions.
%
 Moreover, the major role of contingent payments in these applications
 is to improve revenue as well as to hedge
 risk~\cite{skrzypacz2013auctions}.  In contrast, the role of
 penalties in our setting is two-fold: to provide participants with a
 way to signal their own, idiosyncratic uncertainty, as well as to
 address problems of moral hazard that arise once a resource has been
 assigned.

Our problem is a principal-agent problem~\cite{hart1986theory,holmstrom1979moral}.
Classically, the
principal-agency literature
 addresses both
problems with hidden information (e.g. seller's quality~\cite{dellarocas2003efficiency})
before the time of contracting, which are termed \emph{adverse selection}, and problems for which information asymmetry arises after
the time of contracting (e.g. shipping a low quality good),
the problem of \emph{moral hazard}.
The distinction between the two settings is blurred in dynamic settings (see~\cite{stole2001lectures,bolton2005contract}) such as the
present one. This is because
there are informational asymmetries 
 both before and after contracting.
 In particular, although agents' actions are fully observable, uncertainty together with participation constraints precludes charging unbounded penalties,
 which is a standard approach when actions are observable in moral hazard problems.
 We are not aware of any 
 model or methods in the principal-agency
 literature that addresses our problem.

In regard to auctions in which actions take place after the time of
contracting, \citet{atakan2014auctions} study
auctions where the value of taking each action depends on the
collective actions by others, but these actions are taken before
rather than after observing the world state, and thus the timing of
information is quite different than in our model.  A classical paper
is \citet{courty2000sequential}, who study the
problem of revenue maximization in selling airline tickets, where
passengers have uncertainty about their value for a trip at the time
of booking, and decide whether to take a trip only after realizing
their actual values. Although~\citet{courty2000sequential} model
agents' types as distributions, and the optimal mechanism in their
setting can be understood as a menu of contingent contracts, the type
spaces in their model are effectively one-dimensional, since they require the value distributions to satisfy either stochastic dominance or the ``mean preserving spread" condition. 
%
In both cases, agents' expected utility functions do not cross with each other. We do not impose such constraints on the type space, and one of the major technical difficulties in te present work is the heterogeneous preference of agents over different payment schedules. 

The closest related work is on the design of mechanisms for
incentivizing reliability in the specific setting of demand-side
response in electric power systems~\cite{Ma_ijcai16,Ma_aamas17}, where
selected agents decide whether to respond only after uncertainty in
their costs for demand response are resolved.  The objective there is
to guarantee a probabilistic target on the collective actions taken by
agents, without selecting too large a number of agents or incurring
too much of a total cost. Crucially, and in stark contrast to the
models in the present paper, there is no hard feasibility constraint
in these settings of demand response--- that is, whereas only one
agent can be assigned to a resource in our model, in demand response
problems any number of agents can reduce demand. This additional
feasible constraint has the effect of making the present problem more
challenging.

Other papers study assignment problems under uncertainty, including
models with the possibility that workers assigned to tasks will prove
to be unreliable~\cite{porter2008fault}, and general models of dynamic
mechanism design, where the goal is to maximize expected total
(discounted) value in the presence of
uncertainty~\cite{parkes03c,cavallo10,bergemann2010dynamic} (dynamic
VCG mechnanisms).  What is different between these models and our
problem is that there is no need for the ``shaping'' of downstream
behavior through contingent payments. In~\citet{porter2008fault}, for
example, the probability that an assigned worker fails to complete a
task is fixed.  The solution suggested by dynamic VCG would be to
simply run a second-price auction with reserve $-\socialV$ in the second
period (this auction would allow payments of up to $\socialV$ to agents).
This is outside our design space: we seek mechanisms that assign the
resource in a period {\em before} the value of agents are realized.
We think this is important in the aforementioned motivating settings,
because it allows for planning by agents.

There is also a literature on strictly-proper
scoring-rules~\cite{gneiting2007strictly}, but this does not not
appear to helpful for eliciting the information about uncertainty in
the present context because (i) only the actions, and not realized
values are observed, and thus a scoring-rule method could not be used
to elicit beliefs about value distributions, and (ii) the utility for
using an assigned resource is entangled with the incentives to provide
accurate prediction about one's utilization action.

\section{Preliminaries} \label{sec:preliminaries}
We first introduce the model for the assignment of a single resource. 
There is a set of agents $N = \{1,~2,~\dots, ~n \}$ and two time
periods. In period $T = 0$, the value of each agent $i \in N$
for using the resources is uncertain, represented by a random variable
$V_i \in \setR$, whose exact (and potentially negative) value is not
realized until period $T = 1$ (the time line is more formally presented in the next subsection).
The cumulative distribution function (CDF) $\CDF_i$ of $V_i$ is agent $i$'s private information at period $0$, and corresponds to her \emph{type}. Let $\CDF = (\CDF_1, \dots, \CDF_n)$ denote a type profile.

The assignment is determined in period $0$, whereas the allocated
agent decides on whether to use the resource at period $1$, after she
privately learns the realization $v_i$ of $V_i$.
In addition, if the resource is utilized, then society gains value $\socialV \geq 0$. Define $V^+_i \triangleq \max \{V_{i}, \; 0\}$. We make the following assumptions about $F_i$ for each $i \in N$:

\vsq{-0.5em}

\begin{enumerate}[({A}1)]
	\item $\E{V_{i}^+}>0$, which means that $V_i$ takes positive value with non-zero probability, thus the \emph{option} to use the resource as one wishes has positive value. An agent for which this is violated would never be interested in the resource.
	 
	\item $\E{V_{i}^+} < +\infty$, which means that agents do not get  infinite expected utility from the option to use the resource, thus would not be willing to pay an underhandedly large payment for it. 
\end{enumerate}

Example~\ref{ex:vipi} is a value distribution with discrete support,
and models the type of an agent who may be unable to use the resource. 
Example~\ref{ex:exp_model} is an example type model where values are continuously distributed. Our results do not depend on any assumptions on continuity. 

\begin{example}[$(\fixedV_i,\fixedP_i)$ model] \label{ex:vipi} 
w.The value for agent $i$ to use the resource is $\fixedV_i > 0$, however, she is able to do so only with probability $\fixedP_i  \in (0,1)$. With probability $1 - \fixedP_i$, agent $i$ is unable to show up to use the resource. 
This hard constraint can be modeled as $V_i$ taking value $-\infty$ with probability $1-\fixedP_i$.
See Figure~\ref{fig:pmf_vipi}.  We have $\E{V_i^+} = \fixedV_i  \fixedP_i > 0$. 
If the resource is allocated to agent~$i$, it will be used with probability $\fixedP_i$, and the expected social welfare is $\fixedP_i(\fixedV_i+\socialV)$.
\end{example}

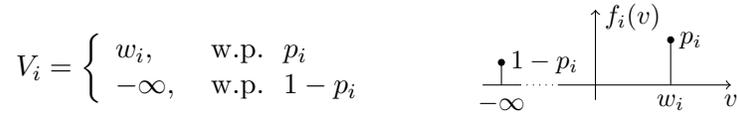
\begin{figure}[t!]
\vsq{-0.5em}
\centering
{
\subfloat{
\small{
\begin{tikzpicture}[scale = 0.9][font=\normalsize]
	\draw (0,0) node  {$V_i = \pwfun{\fixedV_i, &\txtwp\ \fixedP_i \\
		-\infty, &\txtwp\ 1 - \fixedP_i}$};		
	\draw (0,-0.5) node{{\color{white} some text}};
\end{tikzpicture}	
}}
\hspace{2em}
\subfloat{
\begin{tikzpicture}[scale = 1][font = \small]
\draw[->] 	(0.5,0) -- (2.8,0) node[anchor=north] {$v$};
\draw[->] 	(1,-0.2) -- (1,1);
\draw (1,0.9)  node[anchor=west] {$f_i(v)$};
\draw[dotted] (0,0) -- (0.5,0);
\draw (-0.5, 0) -- (0, 0);

\draw (2, 0) -- (2,0.6);
\draw [fill] (2, 0.6) circle [radius=0.04] node[anchor=west] {\small{$\fixedP_i$}};

\draw (-0.25, 0) -- (-0.25,0.3);
\draw [fill] (-0.25, 0.3) circle [radius=0.04]node[anchor=west] {\small{$1-\fixedP_i$}};

\draw	(2, 0) node[anchor= north] {\small{$\fixedV_i$}}
		(-0.25, 0) node[anchor = north] {\small{$-\infty$}};

\end{tikzpicture}
}
}
\caption{Agent value distribution under the $(\fixedV_i,~ \fixedP_i)$ type model.
\label{fig:pmf_vipi}} \vsq{-0.5em}
\end{figure}

\begin{example}[Exponential model]  \label{ex:exp_model} 
The utility for agent $i$ to use the resource is a fixed value $\fixedV_i > 0$ minus a random opportunity cost, which is exponentially distributed with parameter $\lambda_i>0$. 
%
See Figure~\ref{fig:pdf_exp}. The expected value of the random value $V_i$ is $\E{V_i} = \fixedV_i - 1/\lambda_i$, where $1/\lambda_i$ is the expected value of the opportunity cost. 

\end{example}

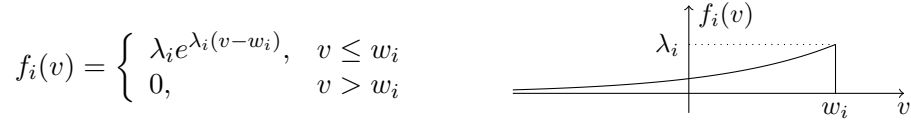
\begin{figure}[t!]
\centering
{
\subfloat{
\small{
\begin{tikzpicture}[scale = 0.9][font=\normalsize]
	\draw (0,0) node  {$	f_i(v) = \pwfun{ \lambda_i e^{\lambda_i (v - \fixedV_i)}, & v \leq \fixedV_i \\
	0, & v > \fixedV_i}$};		
	\draw (0,-0.6) node{{\color{white} some text}};
\end{tikzpicture}	
}}
\hspace{2em}
\subfloat{
\begin{tikzpicture}[scale = 1.3][font=\small]

\draw[->] (-1.8,0.3) -- (2.2,0.3)  node[anchor=north] {$v$};
  
\draw[->] (0,0.1) -- (0,1.2);
\draw (0,1.1) node[anchor=west] {$f_i(v)$};

\draw[scale=1,domain=-1.8:1.5,smooth,variable=\x] plot (\x,{0.5*exp(0.8*(\x-1.5))+0.3});

\draw[-] (1.5, 0.3) -- (1.5,0.8);

\draw (1.5, 0.3) node[anchor=north] {$\fixedV_i$};
\draw[dotted] (0, 0.8) -- (1.5, 0.8);
\draw (0, 0.8) node[anchor=east] {$\lambda_i$};
\end{tikzpicture}
}
}
\caption{Agent value distribution under the exponential type model.
\label{fig:pdf_exp}} \vsq{-0.5em}
\end{figure}



\subsection{Two-Period Mechanisms}

A {\em two-period mechanism} is defined by $\mech = (\reportSet, x, \tzero, \tone)$. At period 0, 
each agent makes a report $\report_i$ from some set of messages $\reportSet$. Let $\report = (\report_1, \dots, \report_n)\in \reportSet^n$ denote a report profile. Based on the reports, an {\em allocation rule} $x = (x_1, \dots, x_n): \reportSet^n \rightarrow \{0, 1\}^n$ assigns the right to use the resource to at most one agent, which we denote as $\winner$, for whom $x_{\winner}(\report) = 1$. $x_i(\report) = 0$ for all $i \neq \winner$. Each agent is charged $\tzero_i(\report)$ in period 0. The mechanism also determines the penalty $\tone_{\winner}(\report)$ for the allocated agent (denote $\tone_i(\report) = 0$ for all $i \neq \winner$).\footnote{More generally, we may think of mechanisms that charge\ the allocated agent a non-zero payment in period $1$ even if she used the resource. Without temporal preference for money, it is without loss to move this part of the payment to period $0$, and at the same time subtract the same amount from the penalty payment.}
%
%
%
The timeline of a two-period mechanism is as follows: \smallskip

\noindent {\em Period~$0$:} \vspace{-0.5em}
\begin{enumerate}[$\bullet$]
	\setlength\itemsep{0em}
	\item Each agent $i \in N$ reports $\report_i \in \reportSet$ to the mechanism based on the knowledge of her type $\CDF_i$.
	\item The mechanism allocates the resource to agent $\winner \in N $, thus $x_\winner(\report) = 1$ and $x_i(\report)= 0$ for $i \neq \winner$. 
	\item The mechanism collects $\tzero_i(\report)$ from each agent, and determines the penalty $\tone_{\winner}(\report)$  for $\winner$.
\end{enumerate}

\vspace{-0.3em}

\noindent {\em Period~$1$}: \vspace{-0.5em}
\begin{enumerate}[$\bullet$]
	\setlength\itemsep{0em}
	\item The allocated agent privately observes the realized value $v_{\winner}$ of $V_\winner$.
	\item The allocated agent decides whether to use the resource  based on $v_{\winner}$ and $\tone_\winner(\report)$.
	\item The mechanism collects the penalty  $\tone_\winner(\report)$ from agent $\winner$ if she did not use the resource.
\end{enumerate}

\begin{example}[Second price auction] The standard second price (SP) auction can be described as a two-period mechanism, where $\mathcal{R} =\mathbb R$, $\winner \in \arg \max_{i \in N} \report_i$, $\tzero_\winner(\report) = \max_{i\neq \winner} \report_i$, and all other payments are 0. The second price auction does not make use of the period~1 payments.
%
\end{example}

\begin{example}[Contingent second price mechanism] The \emph{contingent second price} (CSP) mechanism collects a single bid from each agent, allocates the right to use resource to the highest bidder, and charges the second highest bid, \emph{but only if the allocated agent fails to use the resource}.  Formally, $\mathcal{R} =\mathbb R$, $\winner \in \arg \max_{i \in N} \report_i$, $\tone_\winner(\report) = \max_{i\neq \winner} \report_i$, and all other payments are 0.
\end{example}

\medskip

We assume that agents are risk-neutral, expected-utility maximizers with quasi-linear utility functions. Assume agent $i$ is allocated the resource and is facing a \emph{two part payment} $(z,y)$, where $z$ is the period~$1$ \emph{penalty} payment and $y$ is the period~$0$ \emph{base payment}. Her utility from using the resource in period $1$ is $v_i - y$, and her utility from not using the resource is $-y - z$. Therefore, after observing $v_i$ in period $1$, the rational decision is to use the resource if and only if $v_i - y \geq -y - z \Leftrightarrow v_i \geq - z$ (breaking ties in favor of using the resource). Define $u_i(z)$ as 
\begin{align}
	u_i(z) \triangleq \E{V_i \one{V_i \geq -z}} - z \Pm{V_i < -z} = \E{\max \{V_i, -z\}}, \label{equ:exp_util_z}
\end{align}
where $\one{\cdot}$ is the indicator function, we know that the expected utility of an allocated agent facing two-part payment $(z,y)$ is $u(z) - y$. Under a two-period mechanism, given report profile $\report$, agent $i$'s expected utility is therefore $x_i(\report) u_i(\tone_i(\report)) - \tzero_i(\report)$.

Throughout the paper, we assume that agents make rational decisions in period $1$, if allocated. The interesting question is in agents' incentives regarding their reports in period~$0$.
For any vector $s = (s_1, \dots, s_n)$ and any $i \in N$, we denote $s_{-i} \triangleq (s_1, \dots, s_{i-1}, s_{i+1}, \dots, s_n)$.

\begin{definition}[Dominant strategy equilibrium] 
A two-period mechanism has a {\em dominant strategy equilibrium} (DSE) if for each agent $i \in N$, for any type $\CDF_i$ satisfying (A1) and (A2), there exists a report $\report^\ast_i \in \reportSet$ such that $\forall \report_{i} \in \reportSet, ~\forall \report_{-i} \in \reportSet^{n-1}$, 
\begin{align*}
	x_i(\report^\ast_i, ~ \report_{-i}) u_i(\tone_i(\report^\ast_i, ~ \report_{-i})) - \tzero_i(\report^\ast_i, ~ \report_{-i}) \geq x_i(\report_i, ~ \report_{-i}) u_i(\tone_i(\report_i, ~ \report_{-i})) - \tzero_i(\report_i, ~ \report_{-i}).
\end{align*}
\end{definition}


%
Let $\report^\ast(\CDF) = (\report^\ast_1, \dots, \report^\ast_n)$ denote the report profile under a DSE given type profile $\CDF$.
\begin{definition}[Individual rationality] 
A two-period mechanism is {\em individually rational} (IR) if for each agent $\forall i \in N$, for any type $\CDF_i$ satisfying (A1) and (A2), and any report profile $\report_{-i} \in \reportSet^{n-1}$, 
\begin{align*}
	x_i(\report^\ast_i, ~ \report_{-i}) u_i(\tone_i(\report^\ast_i, ~ \report_{-i})) - \tzero_i(\report^\ast_i, ~ \report_{-i}) \geq 0.
\end{align*}
\end{definition}

IR requires that an agent's expected utility is non-negative under her dominant strategy given that she makes rational decisions in period $1$ (if allocated), regardless of the reports made by the rest of the agents.
IR is based on the expected utility before uncertainty is resolved. It is still possible for an agent to get negative utility at the end of period 1. 
We cannot charge unallocated agents without violating IR, thus $\tzero_i(\report) \leq 0$ for all $i\neq \winner$ for all $\report \in \reportSet^n$.

The expected revenue of a two-period mechanism $\mech$ is the total expected payment from the agents to the mechanism in DSE, assuming rational decisions of agents in period $1$:
\begin{align}
	\rev_\mech(\CDF) \triangleq &
	\sum_{i \in N } \tzero_i(\report^\ast) + 
	\tone_{\winner}(\report^\ast) \cdot \Pm{ V_{\winner} < - \tone_{\winner}(\report^\ast) }.	\label{equ:rev}
\end{align}

\begin{definition}[No deficit] A two-period mechanism satisfies {\em no deficit} (ND) if, for any type profile $\CDF$ that satisfies (A1) and (A2), the expected revenue is non-negative: $\rev_\mech(\CDF) \geq 0$.
\end{definition}

We also consider two additional properties:
A mechanism is \emph{anonymous} if the outcome (assignment, payments) is invariant to permuting the identities of agents. 
A mechanism is \emph{deterministic} if the outcome is not randomized unless there is a tie. 
%
%

The \emph{utilization} achieved by mechanism $\mech$ in dominant strategy is the probability with which the allocated agent rationally decides to use the resource:
\begin{align}
	\ut_\mech(\CDF) \triangleq \Pm{V_{\winner} \geq - \tone_\winner (\report^\ast)}. 
\end{align}

The expected welfare gain to society from the resource being utilized is therefore $\ut_\mech(\CDF)\socialV$, and the expected \emph{social welfare} is the sum of this welfare gain, and the expected value of the agent from using this resource:
\begin{align}
	\sw_\mech(\CDF) \triangleq \E{ V_\winner \one{V_\winner \geq - \tone_\winner (\report^\ast)} } + \socialV \ut_\mech(\CDF).
\end{align}

Our objective is to design mechanisms that maximize
expected social welfare. We do not consider monetary transfers in the social welfare function. The reason $\tone_i(\report^\ast)$ appears is that it affects the decision of the allocated agent in period $1$.

\section{Contingent Payment Mechanism} \label{sec:cont_payment_mech}

We introduce in this section a class of contingent payment mechanisms parametrized by a maximum penalty $\maxZ$ an agent may be charged in period~$1$, and show that under (A1) and (A2), the contingent payment mechanism with $\maxZ = \socialV$ achieves higher welfare and utilization than the second price auction in dominant strategy equilibrium. The uniqueness and optimality are discussed in Section~\ref{sec:optimality}.

\begin{definition}[Contingent payment mechanism] \label{def:cpm}
The \emph{contingent payment mechanism} with maximum penalty $\maxZ$ (the CP($\maxZ$) mechanism) collects two-part bids $b = (b_1, \dots, b_n)$. For each $i \in N$, $b_i = (\bone_i,\bzero_i) \in \reportSet$, where  $\reportSet = \{(z,y) \in \setR^2 ~|~0 \leq z \leq \maxZ,~y = 0 \} \cup \{(z,y) \in \setR^2 ~|~ z = \maxZ, ~ y \geq 0\}$. 
%
%
\begin{enumerate}[$\bullet$]
	\setlength\itemsep{0em}
	\item Allocation rule: $x_{\winner}(b) = 1$ for $ \winner \in \arg \max_{i\in N} \{ \bzero_i + \bone_i\}$ (breaking ties at random). 
	\item Payment rule: let $i' \in \arg\max_{i \neq \winner} \{\bzero_i + \bone_i\}$. $\tzero_\winner(b) = \bzero_{i'}$; $\tone_\winner(b) =\bone_{i'}$; $\tzero_i(b) = 0$, $\forall i \neq \winner$.
\end{enumerate}
\end{definition}

Under the CP($\maxZ$) mechanism, each agent may bid a period~$0$ payment if she is willing to bid a period~$1$ no-show penalty of $\maxZ$, in which case $b_i = (\maxZ,\bzero_i)$ for some $\bzero \geq 0$. Otherwise, she may bid a maximum acceptable penalty (up to $\maxZ$) and no period~$0$ payment, i.e. $b_i = (\bone_i,0)$ for some $\bone_i \in [0, \maxZ]$.
The resource is allocated to the highest period~$0$ payment bidder, if there exist any agent with non-zero $\bzero$ bid (since $\bone_i \leq \maxZ$ thus $\bone_i + 0 \leq \bzero_i + \maxZ$).
Otherwise, the resource is allocated to the highest period~$1$ penalty bidder. The allocated agent is charged a two part payment equal to the bid of the second ``highest" bidder. 

Recall that $u_i(z) = \E{\max\{ V_i, -z\}}$ as defined in \eqref{equ:exp_util_z} is the expected utility of agent $i$ if she were allocated and charged only a penalty $z$--- in period $1$, the agent gets either her realized value $v_i$, or get $-z$ from paying the penalty, whichever is higher. $u_i(z)$, as a result, is also the highest base payment agent $i$ is willing to accept, when her penalty is $z$. We first state some useful properties.

\begin{restatable}{lemma}{lemmaExpUtility} \label{lem:exp_u}  
Assuming (A2), the expected utility $u_i(z)$ as a function of the penalty $z$ satisfies:
\begin{enumerate}[(i)]
	\setlength\itemsep{0em}
	\item $u_i(0) = \E{V_i^+}$, $\lim_{z \rightarrow +\infty} u_i(z) = \E{V_i}$.
	\item $u_i(z)$ is continuous, convex and monotonically decreasing with respect to $z$.	
\end{enumerate}
\end{restatable}

See Figure~\ref{fig:utilities}. The proof of this lemma is straightforward. For part (i), $u_i(0) = \mathbb{E}[V_i^+]$ by definition, and $\lim_{z \rightarrow \infty} \E{u_i(z)} = \E{V_i}$ holds by the monotone convergence theorem.
Part (ii) holds since $\max\{V_i, -z\}$ is monotonically non-increasing in $z$, continuous, and convex, so $u_i(z)$ inherits these properties. 
Intuitively, when $z = 0$, the agent uses the resource if and only if her realized value is non-negative thus gets expected utility $\mathbb{E}[V_i^+]$. As the penalty $z$ increases, the agent's expected utility continuously decreases. When $z = \infty$, the agent always uses the resource and never pays the penalty thus her expected utility converges to $\E{V}$. 

\begin{figure}[t!]
\centering 
\subfloat[\small{$u_i(\maxZ) \geq 0$}]{\label{fig:utilieis_posEv}
\begin{tikzpicture}[scale = 1.1][font=\small]
\draw[->] 	(-0.1,0) -- (3.2,0) node[anchor=north] {$z$};
\draw[->] 	(0,-0.2) -- (0,1) node[anchor=west] {$u_i(z)$};

\draw[-] 	(0, 0.8) parabola[bend at end] (3, -0.05); 

\draw[dotted] 	(-0.05,0.23) -- (1.55, 0.23);
\draw[dotted] 	(1.3,-0.1) -- (1.3, 0.4);
\draw (1.3, -0.1) node[anchor=north] {\footnotesize{$\maxZ$}};

\draw	(0, 0.8) node[anchor=east] {\footnotesize{$\mathbb{E}[V_i^+]$}};
\draw	(0, 0.22) node[anchor=east] {\footnotesize{$u_i(\maxZ)$}};
\draw (2.2, 0) node[anchor=north] {\footnotesize{$\zc_i$}};

\end{tikzpicture}
}
\hspace{1em}
\subfloat[$u_i(\maxZ) < 0$]{\label{fig:utilities_negEv}

\begin{tikzpicture}[scale = 1.1][font=\small]
\draw[->] 	(-0.1,0) -- (3.2,0) node[anchor=north] {$z$};
\draw[->] 	(0,-0.2) -- (0, 1) node[anchor=west] {$u_i(z)$};

\draw[-] 	(0, 0.8) parabola[bend at end] (2, -0.2) -- (3, -0.2);

\draw	(0, 0.6) node[anchor=east] {\scriptsize{$\mathbb{E}[V_i^+]$}};
\draw	(1,0) node[anchor = north] {\footnotesize{$\zc_i$}};

\draw[dotted] 	(1.6,-0.1) -- (1.6, 0.4);
\draw (1.6, -0.1) node[anchor=north] {\footnotesize{$\maxZ$}};

\end{tikzpicture}

}
\caption{Expected utility for being allocated the resource
as a function of penalty $z$. 	
\label{fig:utilities}} 
\vsq{-0.5em}
\end{figure}
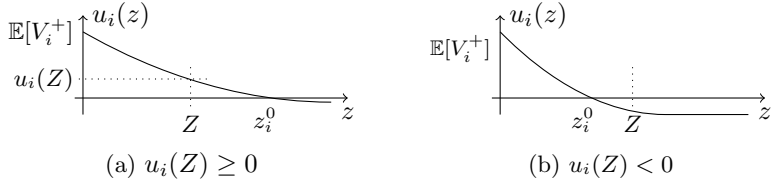

\begin{theorem}[Dominant Strategy in CP($\maxZ)$] 
\label{thm:dominant_strategy} Given~(A1)-(A2), under the 
CP(\maxZ) mechanism, it is a dominant strategy for each agent $i \in N$ to bid $b^\ast_{i, \txtCPM} = (\maxZ, u_i(\maxZ))$ if $u_i(\maxZ) \geq 0$. Otherwise, it is a dominant strategy to bid $b^\ast_{i, \txtCPM} = (\zc_i, 0)$, where $\zc_i$ is the unique zero-crossing of $u_i(z)$.
\end{theorem}

\begin{proof}

First, observe that the message space $\reportSet$ is effectively one-dimensional. For any pair of bids $(\bone_i, \bzero_i), ~(\tilde{b}\1_i, \tilde{b}\0_i) \in \reportSet$, denote $(\bone_i, \bzero_i) \succeq (\tilde{b}\1_i, \tilde{b}\0_i)$ if $\bone_i+ \bzero_i \geq \tilde{b}\1_i + \tilde{b}\0_i$. For any type $F_i$ satisfying (A1) and (A2), the agent's expected utility is weakly lower for a higher payment, i.e. $\bone_i + \bzero_i \geq  \tilde{b}\1_i  + \tilde{b}\0_i\Rightarrow u_i(\bone_i) - \bzero_i \leq u_i(\tilde{b}\1_i) - \tilde{b}\0_i$.

We now show that for any agent, her utility at the two-part bid $b^\ast_{i, \txtCPM}$ is exactly zero. If $u_i(\maxZ) \geq 0$, the agent gets expected utility $u_i(\maxZ) - u_i(\maxZ) = 0$ if she is charged $b^\ast_{i, \txtCPM} = (\maxZ,u_i(\maxZ))$. If $u_i(\maxZ) < 0$, the continuity, convexity and monotonicity of $u_i(z)$ (part (ii) of Lemma~\ref{lem:exp_u}) implies that there is a unique zero crossing $\zc_i < \maxZ$ of $u_i(z)$ s.t. $u_i(\zc_i) = 0$. If she is charged  $b^\ast_{i, \txtCPM} = (\zc_i, 0)$, her expected utility is then $u_i(\zc_i) - 0 = 0$.
This implies that the bids $b^\ast_{i, \txtCPM}$ is an agent's ``highest acceptable payment" in the message space $\reportSet$. The argument for DSE is then standard, observing that the mechanism allocates to the highest bidder and charges a second highest bid.
\end{proof}

Intuitively, under the CP$(\maxZ)$ mechanism, it is a dominant strategy for each agent to bid the additional amount she is willing to pay at period~$0$, given a period~$1$ penalty $\maxZ$, otherwise, the dominant strategy is to bid her highest acceptable penalty when there is no period~$0$ payment.
When $\maxZ = 0$, the CP($\maxZ$) mechanism reduces to SP, where it is a dominant strategy to bid $b_{i,\txtSP}^\ast = u_i(0) = \mathbb{E}[V_i^+]$. 
When $\maxZ \rightarrow +\infty$, and with the additional assumption,

\begin{enumerate}[(A3)]
	\setlength\itemsep{0em}
	\item $\E{V_i} < 0$, meaning that being forced to always use the resource is not favorable, 
\end{enumerate}
%
then CP($\maxZ$) reduces to the CSP mechanism, where it is a dominant strategy to bid the largest acceptable penalty $\zc_i$, the unique zero crossing of $u_i(z)$ (see Figure~\ref{fig:utilities}). $\zc_i$ exists and is unique given (A3), since $u_i(z)$ is continuous, monotonically decreasing in $z$, and converges to $\E{V_i} < 0$.\footnote{(A3) should not be confused with (A1), which means that the value of the \emph{option} to use the resource is $\E{V_i^+} > 0$. (A3) only requires that an agent gets negative expected utility from committing to \emph{always} use the resource, regardless of what happens. This a very natural assumption: without (A3), an agent would accept any unboundedly large penalty for the right to use a resource.}

\subsection{Better Welfare and Utilization than Second Price Auction}

The following lemma states useful properties of utilization and welfare as functions of penalty $z$.

\begin{restatable} 
{lemma}{lemmaExpUtilityUtilization} \label{lem:util_welfare}
Assuming (A1) and (A2), when agent $i$ is allocated and charged a two part payment $(z,y)$, the utilization and social welfare are independent of the base payment $y$, and satisfy:
\begin{enumerate}[(i)]
	\setlength\itemsep{0em}
	\item the utilization $\ut_i(z)= \Pm{V_i \geq -z}$ is right continuous and monotonically non-decreasing in $z$. Moreover, $\ut_i(z) =  1 + u_i'(z+)$, where $u_i'(z+)$ is right derivative of $u_i$ at $z$. 
	\item the social welfare $\sw_i(z)= \E{V_i \one{V_i \geq -z}} + \socialV \Pm{V_i \geq -z} $ is right continuous, monotonically non-decreasing in $z$ when $z \leq \socialV$, and monotonically non-increasing in $z$ when $z > \socialV$.
\end{enumerate}
\end{restatable}

See Appendix~\ref{appx:proof_lem_util_welfare} for the proof. The continuity and monotonicity of $\ut_i(z)$ is obvious. From Fubini's theorem, we get $u_i(z) = \mathbb{E}[V_i^+] - \int_0^{z} F_i(-v) dv$ when $z \geq 0$ and $u_i(z) = -z + \int_{-z}^\infty \Pm{V_i \geq v}dv$ when $z < 0$. By the fundamental theorem of calculus, the right derivative of $u_i(z)$ is equal to the right limit of $-F_i(-v)$ at $z$, which is $\Pm{V_i \geq -z} - 1$. For part (ii), observe that $\sw_i(z) = \sw_i(z) = \E{(V_i+\socialV) \one{V_i \geq -z}}$, and that the random variable $V_i + \socialV$ is non-negative iff $V_i \geq - \socialV$. 

Intuitively, the agent uses the resource with higher probability when the penalty $z$ increases. This, in turn, results in a smaller probability of paying the penalty, thus $u_i(z)$ decreases \emph{slower} as $z$ increases, corresponding to a shallower slope of the convex function $u_i(z)$. The welfare-optimal utilization decision in period~1 is to use the resource iff the realized value $v_i \geq -\socialV$, therefore $z = \socialV$ optimizes $\sw_i(z)$.
%
%
With Lemma~\ref{lem:util_welfare}, we prove the following result:

\begin{restatable}{lemma}{lemCrossingUtils} \label{lem:crossing_utilities}
Let $u_1(z)$ and $u_2(z)$ be the expected utilities of two agents whose types satisfy (A1) and (A2), and consider $z_1,z_2\in \setR$ s.t. $z_1 < z_2$. If $u_1(z_1) \geq u_2(z_1)$, and $u_1(z_2) \leq u_2(z_2)$, we have:
\begin{enumerate}[(i)]
	\setlength\itemsep{0em}
	\item $\ut_1(z_1) \leq \ut_2(z_2)$.
	\item $\sw_1(z_1) \leq \sw_2(z_2)$ if $z_1 \leq z_2 \leq \socialV$, 
	and $\sw_1(z_1) \geq \sw_2(z_2)$ if $\socialV \leq z_1 \leq z_2$.
\end{enumerate}
\end{restatable}

When $u_2(z)$ crosses $u_1(z)$ from below, $u_1(z_2) - u_1(z_1) \leq u_2(z_2) - u_2(z_1)$. 
The convexity of $u_i (z)$ then implies that the right derivative of $u_2(z)$ at $z_2$ must be  higher than the right derivative of $u_1(z)$ at $z_1$, hence the inequality on utilization. See Appendix~\ref{appx:proof_lem_crossing_utilities} for details, and the proof for part (ii).   
The CP mechanism with the maximum penalty set to $\socialV$ will have some very nice optimality properties. As a preliminary observation, we state the following result relative to the SP auction.

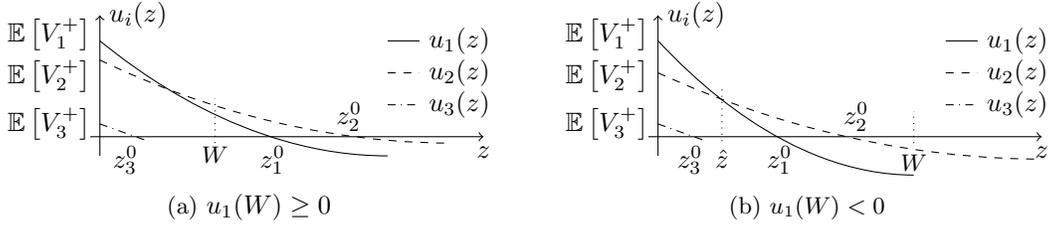
\begin{figure}
\centering 
\subfloat[\small{$u_1(\socialV) \geq 0$}]{\label{fig:u2_pos}
\begin{tikzpicture}[scale = 0.85][font=\small]

\draw[->] (-0.15,0) -- (6,0) node[anchor=north] {$z$};

\draw[->] (0,-0.3) -- (0, 1.9) node[anchor=west] {$u_i(z)$};

\draw[-] (0, 1.5) parabola[bend at end] (4.5,-0.3);

\draw[dashed] (0,1.2) parabola[bend at end] (5.5, -0.1);

\draw[dashdotted] (0, 0.2) -- (0.7, -0.05);

\draw	
		(2.7,0) node[anchor=north] {$\zc_1$}
		(3.9,-0.05) node[anchor = south] {$\zc_2$}
		(0.4, 0) node[anchor=north] {$\zc_3$};
		
\draw	(0, 1.6) node[anchor=east] {\small{$\E{V_1^+}$}}
		(0, 0.95) node[anchor=east] {\small{$\E{V_2^+}$}}
		(0, 0.2) node[anchor=east] {\small{$\E{V_3^+}$}};

\draw[dotted] 	(1.8,-0.1) -- (1.8, 0.6);
\draw (1.8, -0) node[anchor=north] {\footnotesize{$\socialV$}};

\draw (4.5, 1.5) -- (5, 1.5) node[anchor=west] {$u_1(z)$};
\draw[dashed] (4.5, 1) -- (5, 1) node[anchor=west] {$u_2(z)$};
\draw[dashdotted] (4.5, 0.5) -- (5, 0.5) node[anchor=west] {$u_3(z)$};
\end{tikzpicture}
}
\hspace{1em}
\subfloat[$u_1(\socialV) < 0$]{\label{fig:u2_neg}

\begin{tikzpicture}[scale = 0.85][font=\small]

\draw[->] (-0.15,0) -- (6,0) node[anchor=north] {$z$};

\draw[->] (0,-0.3) -- (0, 1.9) node[anchor=west] {$u_i(z)$};

\draw[-] (0, 1.5) parabola[bend at end] (4,-0.6);

\draw[dashed] (0,1) parabola[bend at end] (6, -0.35);

\draw[dashdotted] (0, 0.2) -- (0.7, -0.05);

\draw	
		(1.9,0) node[anchor=north] {$\zc_1$}
		(3.1,-0.05) node[anchor=south] {$\zc_2$}
		(0.5,0) node[anchor=north] {$\zc_3$};
		
\draw	(0, 1.6) node[anchor=east] {\small{$\E{V_1^+}$}}
		(0, 0.95) node[anchor=east] {\small{$\E{V_2^+}$}}
		(0, 0.2) node[anchor=east] {\small{$\E{V_3^+}$}};
		
\draw[dotted] (1,-0.07) -- (1, 0.8);
\draw (1, -0.06) node[anchor=north] {$\hat{z}$};	
		
\draw (4.5, 1.5) -- (5, 1.5) node[anchor=west] {$u_1(z)$};
\draw[dashed] (4.5, 1) -- (5, 1) node[anchor=west] {$u_2(z)$};
\draw[dashdotted] (4.5, 0.5) -- (5, 0.5) node[anchor=west] {$u_3(z)$};

\draw[dotted] 	(4,-0.1) -- (4, 0.4);
\draw (4, -0.13) node[anchor=north] {\footnotesize{$\socialV$}};

\end{tikzpicture}
}
\caption{Economies where the SP winner is different from the CP($\socialV$) winner.
 \label{fig:switched_bids}} \vsq{-0.5em}
\vsq{-0.5em}
\end{figure}

\begin{restatable}{theorem}{thmCSPbeatSP} \label{thm:CSP_beat_SP}
For any set of agent types satisfying (A1)-(A2), under the dominant strategy equilibria, the CP$(\socialV)$ mechanism mechanism Pareto-dominates the SP auction in utilization and welfare.
%
\end{restatable}
\begin{proof} Consider the following two cases:

\noindent{}\emph{Case 1.} SP and CP($\socialV$) allocate the resource to the same agent. Assume that CP($\socialV$) charges penalty $z^\ast$, we know $z^\ast \in [0, \socialV]$. The utilization and welfare under CP($\socialV$) are always (weakly) higher than those under SP, given the monotonicity properties proved in Lemma~\ref{lem:util_welfare}.

\noindent{}\emph{Case 2.} SP and CP($\socialV$) allocate the resource to agent 1 and 2 respectively.  We know that for agent $2$ to be allocated under CP($\socialV$), either $u_1(\socialV) \geq 0$, in which case $\socialV = \min \{\zc_1, \socialV\}$ and $u_2(\socialV) \geq u_1(\socialV)$ (Figure~\ref{fig:u2_pos}), or $u_1(\socialV) < 0$, in which case $\zc_1 = \min\{ \zc_1, \socialV \}$ and $u_2(\zc_1) \geq u_1(\zc_1)$ (Figure~\ref{fig:u2_neg}). In both cases, $u_2(\min\{\zc_1,\socialV\}) \geq u_1(\min\{\zc_1,\socialV\})$ holds, and we have $z^\ast \in[ \min\{\zc_1,\socialV\}, \socialV]$, where $z^\ast$ is the penalty that agent $2$ is charged under CP($\socialV$). Given that agent $1$ is the SP winner we also know that  $u_1(0) \geq u_2(0)$. Applying Lemma~\ref{lem:crossing_utilities}, we know $ut_2(z^\ast) \geq ut_2(\min\{\zc_1,\socialV\}) \geq \ut_1(0)$ and $\sw_2(z^\ast) \geq \sw_2(\min\{\zc_1,\socialV\}) \geq \sw_1(0)$, i.e. CP$(\socialV)$ achieves better welfare and utilization. 
\end{proof}

%
This domination result holds for arbitrary tie-breaking rules for the two mechanisms if $\socialV > 0$. 
The same analysis on the CSP mechanism shows that it always achieves a higher utilization than the second price auction.
We illustrate through the following examples the improvement in welfare and utilization from CP($\socialV$) over SP, and show that SP can be arbitrarily worse than CP($\socialV$).

\begin{example} [Double gain in CP($\socialV$)] 
\label{ex:3_val} 
Consider $\socialV = 50$, and two agents with value distributions and expected utility functions as shown in Figure~\ref{fig:example_utilities}. Compared with agent~2, agent~1 has higher value for the resource, but lower probability of willing to use the resource and higher probability for a hard constraint.
Under SP, the DSE bids are $b_{1, \txtSP}^\ast = 20$, $b_{2, \txtSP}^\ast  = 16$ thus agent~1 is allocated. The utilization is $\ut_1(0) = \Pm{V_1 \geq 0} = 0.2$ and the welfare is $\sw_1(0) = 100 * 0.2 + 50 * 0.2 = 30$.

Whereas under CP($\socialV$) mechanism, $b_{1, \txtCPM}^\ast = (\zc_1, 0) = (30, 0)$ and $b_{2, \txtCPM}^\ast = (\socialV, u_2(\socialV)) = (50, 2)$. 
Agent~2 is allocated and charged penalty $\tone_2(b) = \zc_1 = 30$, thus the utilization is $\Pm{V_2 \geq - 30} = 0.8$, and the welfare is $\sw_2(30) = 40*0.4 - 10*0.4 + 50*0.8 = 52$. Note that the these are higher than $\ut_2(0) =  \Pm{V_2 \geq 0} = 0.4$ and $\sw_2(0) = 36$, what is achieved if agent~2 is allocated the resource under SP in some other economy and charged no penalty. 
\qed
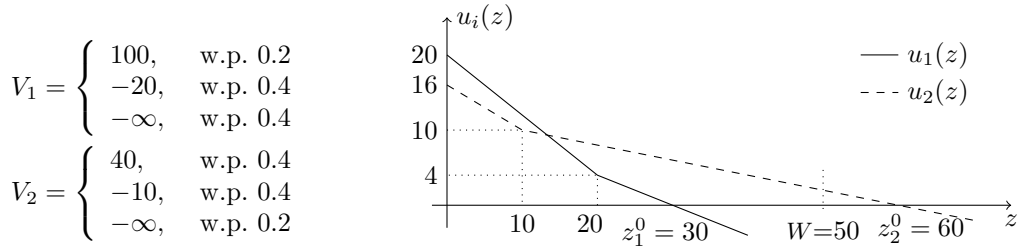
\begin{figure}[t!]
\centering{
\subfloat{
\small{
\begin{tikzpicture}[scale = 1][font=\small]
\draw (0,1.4) node {$V_1 = \pwfun{100, & \txtwp 0.2\\ -20, & \txtwp 0.4 \\ -\infty, &\txtwp 0.4}$};
	\draw (0,0) node  {$V_2 = \pwfun{40, & \txtwp 0.4\\ -10, & \txtwp 0.4 \\ -\infty, &\txtwp 0.2}$};
	\draw (0,-0.6) node{{\color{white} some text}};	
\end{tikzpicture}	
}}
\hspace{2em}
\subfloat{
\begin{tikzpicture}[scale = 1][font=\small]

\draw[->] (-0.2,0) -- (7.5,0) node[anchor=north] {$z$};

\draw[->] (0,-0.3) -- (0,2.5) node[anchor=west] {$u_i(z)$};

\draw[-] (0, 2) -- (2, 0.4) -- (4, -0.4);
\draw[dashed] (0, 1.6) -- (1, 1) -- (7, -0.2);

\draw	
		(2.9, -0.05) node[anchor=north] {$\zc_1 = 30$}
		(6.3, -0) node[anchor=north] {$\zc_2 = 60$}
		(5, -0.1) node[anchor = north] {{\footnotesize $\socialV$}=50}
		(1, 0) node[anchor = north] {10}
		(1.9, 0) node[anchor = north] {20}
		(0, 1) node[anchor = east] {10}
		(0, 0.4) node[anchor = east] {4};
		
\draw	(0, 2.0) node[anchor=east] {$20$}
		(0, 1.6) node[anchor=east] {$16$};
		
\draw (5.5, 2) -- (6,2) node[anchor=west] {$u_1(z)$};
\draw[dashed] (5.5, 1.5) -- (6, 1.5) node[anchor=west] {$u_2(z)$};

\draw[dotted] 	(2,0) -- (2, 0.4) -- (0, 0.4);
\draw[dotted]  (0, 1) -- (1,1) -- (1, 0);
\draw[dotted] 	(5, -0.05) -- (5, 0.5);

\end{tikzpicture}}
}
\vsq{-0.5em}
\caption{Agents' value distributions and expected utilities in Example \ref{ex:3_val}. 
 \label{fig:example_utilities}} \vsq{-0.5em}
\end{figure}
\end{example}

\begin{example}[SP arbitrarily worse] \label{ex:SP_worse} 
Under the $(\fixedV_i, \fixedP_i)$ model introduced in Example~\ref{ex:vipi}, the expected utility for agent $i$ given penalty $z$ is $u_i(z) = \fixedV_i \fixedP_i - (1- \fixedP_i)z$. Consider an economy with two $(\fixedV_i,\fixedP_i)$ agents: $\fixedP_1=\eps$, $\fixedV_{1} = 1/\eps$, and $\fixedP_2 = 1-\eps$, $\fixedV_{2} = 1$ for some very small $\eps > 0$. Agent 1 is allocated under SP since $b_{1, \txtSP}^\ast = 1 > b_{2,\txtSP}^\ast = 1-\eps$, whereas agent 2 is allocated under CP$(\socialV)$ as long as $\socialV > 1$, since $ \zc_1 = 1/(1-\eps) \approx 1$ and $\zc_2 = (1-\eps)/\eps \gg 1$. The utilization under SP and CP$(\socialV)$ are $\eps$ and $1-\eps$, respectively, and the welfare under the two mechanisms are $1 + \eps \socialV$ and $(1-\eps) (1 + \socialV)$. Thus, CP$(\socialV)$ can have arbitrarily better utilization and welfare by selecting a better winner. \qed
\end{example}

The higher welfare and utilization achieved by the CP$(\socialV)$ come from two aspects of its design. First, charging a penalty $z \in [0, \socialV]$
changes the period 1 decision of the allocated agent, promoting the
resource to be used more efficiently.
Second, the CP($\socialV$) mechanism selects a better winner:
\begin{enumerate}[$\bullet$]
	\item For $i$ s.t. $u_i(\socialV) \geq 0$, the two-part bids under CP($\socialV$) add up to $u_i(\socialV) + \socialV = \E{V_i\one{V_i \geq \socialV}}$ $ - \socialV  \Pm{V_i < -\socialV} + \socialV = \E{(V_i + \socialV) \one{V_i \geq -\socialV}} = \sw_i(\socialV)$, the highest achievable welfare from allocating the resource to agent $i$ and setting an optimal penalty $z = \socialV$. As a result, if $\max_{i \in N} \{\bzero_i + \bone_i \} \geq \socialV$, CP($\socialV$) selects the agent with highest achievable welfare. 
	\item When $\socialV$ is large and $u_i(\socialV) < 0$, agents with higher probabilities of showing up have $u_i(z)$ that decrease more slowly with $z$, thus have relatively higher zero crossing $\zc_i$, and are more likely to be allocated. With large $\socialV$, higher utilization is more likely to generate higher social welfare.
\end{enumerate}

\section{Characterization and Optimality of CP} \label{sec:optimality}

In this section, we study the optimal mechanism design problem with the following properties:\footnote{For (P5) deterministic, we require that the outcome is deterministic unless multiple agents make the same reports, and that when breaking ties, the two-part payment each agent may be charged if allocated is still deterministic. We also assume that the mechanism uses minimum tie-breaking, and satisfies the \emph{positive responsiveness} requirement, i.e. if a tied agent was to make a ``strictly higher" report in an otherwise equivalent economy, then she has to be allocated with probability one in this other economy. See Appendix~\ref{appx:proof_P1P5}.}
%
\vspace{-0.5em}
\begin{multicols}{2}
\begin{enumerate}[{P}1.]

	\setlength\itemsep{0em}
	\item Dominant-strategy equilibrium
	\item Individually rational
	\item No deficit 
	\item Anonymous
	\item Deterministic (unless breaking ties)
	\item No payment to unassigned agents
\end{enumerate}
\end{multicols}
\vspace{-0.5em}

\begin{figure}
\centering   
\subfloat[Iso-profit curves of an agent]{\label{fig:indiffCurves}

\begin{tikzpicture}[scale = 1.05][font = \small]
\draw[->] (-0.1,0) -- (4.4,0) node[anchor=north] {$z$};

\draw[->] (0,-0.2) -- (0,2.2) node[anchor=west] {$y$};


\draw[-, name path = uA] (-0.1, 1.2) parabola[bend at end] (4,-0.2);

\draw[dashdotted] (-0.1, 2) parabola[bend at end] (4, 0.5);

\draw[dashed] (-0.1, 0.4)  to [out = -40, in = 160] (0.7, -0.05);

\draw	
		(2.3, 0) node[anchor=north] {{ $b_{i,\txtCSP}^\ast = \zc$}}
		(0, 1.2) node[anchor=east] {$b_{i,\txtSP}^\ast$};

\draw[dotted] (-0.1, 0.55) -- (1.45, 0.55);
\draw (0, 0.6) node[anchor = east]{$\bzero_{i, \txtCPM}$};

\draw[dotted] (1.0,-0.1) -- (1.0, 0.9);
\draw (1.0, 0) node[anchor=north]{{\scriptsize $\maxZ$}};

\draw[dashdotted] (2.3, 2) -- (2.7, 2) node[anchor=west] {{\scriptsize $u_i(z) -y = -10$}};
\draw (2.3, 1.55) -- (2.7, 1.55) node[anchor=west] { {\scriptsize $u_i(z) -y = 0$}};
\draw [dashed](2.3, 1.1) -- (2.7, 1.1) node[anchor=west] {{\scriptsize $u_i(z) -y = 10$}};

\end{tikzpicture}
}
\hspace{1em}
\subfloat[Examples of IR and ND ranges]{\label{fig:irbb}
\begin{tikzpicture}[scale = 1.15][font = \small]

\fill [pattern =north west lines, pattern color = black!40, opacity=0.5](-3, 0.68)--(-0.8, -0.2)--(-0.8, -0.6)--(-3, -0.6)--(-3, 0);	

\fill [pattern=crosshatch dots, pattern color = black!60, opacity=0.5] (-3, 0.08)--(-1.3, -0.6)--(-0.8, -0.6)--(-0.8, 0.8)--(-3, 0.8) -- (-3, 0.08);

\draw[->] 	(-3, 0) -- (-0.6,0) node[anchor=north] {$z$};
\draw[->] 	(-2.8,-0.8) -- (-2.8,1.1) node[anchor=west] {$y$};


\draw[-] (-3, 0.68) -- (-0.8, -0.2);

\draw[dashed] 	(-3, 0.08) -- (-1.3, -0.6);

\draw	(-1.2, 0) node[anchor=south] {$\zc_i$};

\draw	(-1.8, -1.4) node[anchor=south] {(i) $(\fixedV_i,\fixedP_i)$};


\fill [pattern =north west lines, pattern color = black!40, opacity=0.5](-0.1, 0.7) parabola[bend at end] (2, -0.4)--(2, -0.6)--(-0.1, -0.6)--(-0.1, 0);	

\fill [pattern=crosshatch dots, pattern color = black!60, opacity=0.5](-0.1, 0.1) parabola[bend at end] (0.6, -0.4)  to [out = 0, in=-170] (2, -0.1)--(2, 0.8)--(-0.1, 0.8) -- (-0.1, 0.1);

\draw[->] 	(-0.1,0) -- (2.2,0) node[anchor=north] {$z$};
\draw[->] 	(0,-0.8) -- (0,1.1) node[anchor=west] {$y$};

\draw[-] 	(-0.1, 0.7) parabola[bend at end] (2, -0.4);

\draw[dashed] 	(-0.1, 0.1) parabola[bend at end] (0.6, -0.4);
\draw[dashed] (0.6, -0.4)  to [out = 0, in=-170] (2, -0.1);
		
\draw	(0.85,0) node[anchor=south] {$\zc_i$};


\filldraw [black] (1.32, -0.28) circle (1pt);
\draw (1.32, -0.28) node[anchor = north]{{\scriptsize $B$}};

\draw	(1, -1.4) node[anchor=south] {(ii) Exponential};


\draw (2.5, 0.6) -- (2.8, 0.6) node[anchor=west] { {\scriptsize $u_i(z) -y = 0$}};
\draw [dashed](2.5, 0.2) -- (2.8, 0.2) node[anchor=west] { {\scriptsize $rev_i(z,y) = 0$}};
	
\fill [pattern =north west lines, pattern color = black!40, opacity=0.5](2.5, -0.1) -- (2.8, -0.1) -- (2.8, -0.25) -- (2.5, -0.25);
\draw	(2.8,-0.2) node[anchor=west] { {\scriptsize IR range}};
	
\fill [pattern=crosshatch dots, pattern color = black!60, opacity=0.5](2.5, -0.5) -- (2.8, -0.5) -- (2.8, -0.65) -- (2.5, -0.65);
\draw	(2.8, -0.6) node[anchor=west] { {\scriptsize ND range}};	
   
\end{tikzpicture}
}
\caption{Iso-profit curves, IR and ND ranges in the two-dimensional payment space. \label{fig:iso_irbb}} \vsq{-0.5em}
\end{figure}
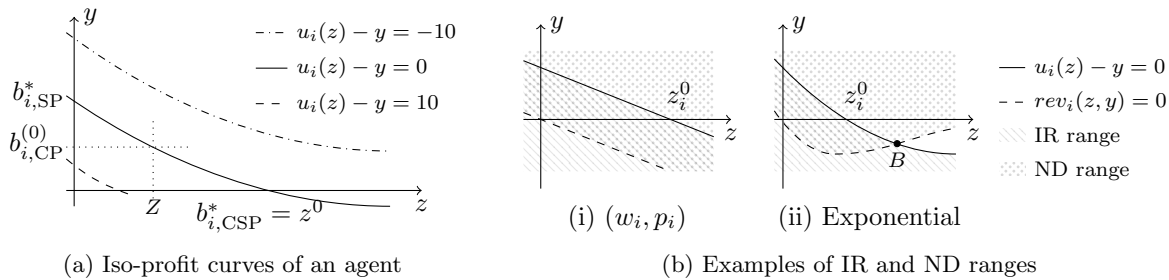

Recall that while facing a two part payment $(z,y)$, an agent's expected utility is $u_i(z) - y$. We work with \emph{iso-profit curves} in the two dimensional payment space, which are sets of $(z,y)$ pairs for which $u_i(z) - y = \alpha $ for some constant $\alpha$, i.e. an agent will be indifferent to all payments $(z,y)$ that reside on the same iso-profit curve. See Figure~\ref{fig:indiffCurves}. 
The \emph{zero-profit curve} (i.e. where $\alpha = 0$, the solid line depicted in Figure~\ref{fig:indiffCurves}) is characterized by $y = u_i(z)$, thus is continuous, convex, monotonically decreasing (Lemma~\ref{lem:exp_u}). Other iso-profit curves are vertical shifts of the zero-profit curve, and recall from Lemma~\ref{lem:util_welfare} that the utilization for an agent facing payments $(z, y)$ relates to the slope of the zero-profit curve: $\ut_i(z) = u_i'(z+)$.

An agent gets negative expected utility if she is allocated and charged a payment $(z,y)$ above her zero-profit curve. We call the area in the $(z,y)$ space weakly below an agent's zero-profit curve the \emph{IR-range} for the agent. For any $(z,y)$, the expected revenue (e.g. payment from the agent to the mechanism) is $rev_i(z, y) = y + z \cdot \Pm{V_i < -z}$.
We call the set of $(z, y)$ for which $rev_i(z, y) \geq 0$ the \emph{ND-range} (no-deficit range) for this agent.

\begin{example} \label{ex:irbb}
Consider the $(\fixedV_i, \fixedP_i)$ model (Example~\ref{ex:vipi}). The zero-profit curve is characterized by 
$y = \fixedV_i \fixedP_i - (1-\fixedP_i)z$. The expected revenue of the mechanism is $rev_i(z,y) = y + (1-\fixedP_i)z$, thus the ND range is lower-bounded by $y = -(1-\fixedP_i)z$. See Figure~\ref{fig:irbb}(i). For an agent with an exponential value distribution as in Example~\ref{ex:exp_model}, the zero-profit curve, IR and ND ranges are as shown in Figure~\ref{fig:irbb}(ii) (see derivations in Appendix~\ref{appx:bids}).
\end{example}

%
Under any two-period mechanism that is IR and ND, the payment facing the assigned agent must reside in the intersection of the IR and ND ranges of the assigned agent. 
Given the monotonicity properties in Lemma~\ref{lem:util_welfare}, the highest possible utilization for an agent subject to IR and ND constraints, i.e. the \emph{first-best utilization}, is achieved by charging a two-part payment with the highest penalty $z$ within the intersection of IR and ND ranges. The first \emph{first best social welfare} is achieved by charging the agent a penalty $z = \socialV$,  or by charging the agent the highest penalty within the intersection of IR and ND ranges (if penalty $ z = \socialV$ is  outside this intersection). See Proposition~\ref{prop:fb_sw} in Appendix~\ref{appx:proof_fb_sw}.

To illustrate this for the exponential type, the first-best utilization is achieved by charging $(z_B, y_B)$ at point $B$ in Figure~\ref{fig:irbb}(ii). This also achieves the first best welfare if $\socialV \geq z_B$, otherwise, the first best welfare is achieved by charging $(\socialV, y)$ for some $y \in [- (1-\ut_i(\socialV))\socialV, u_i(\socialV)]$. 
For the $(\fixedV_i, \fixedP_i)$ type, although there is no upper bound on the highest IR and ND penalty, as long as $z \geq -\fixedV_i$, we have $\ut_i(z) = \fixedP_i$ and $\sw_i(z) = \fixedV_ip_i + \socialV \fixedP_i$, which are not affected by the penalty $z$.

\subsection{Optimality of CP} \label{sec:opt_CPM}

Define the {\em frontier} of a set of agents $N$ with type profile $\CDF = (\CDF_1, \dots, \CDF_n)$ to be the upper-envelope of the zero-profit curves of all agents, i.e. for all $z \in \setR$,
$u_N(z) \triangleq \max_{i \in N} u_i(z)$. This characterizes the maximum willingness to pay (as base payment, given penalty $z$) by all agents in $N$. 
As the upper envelope of a finite set of continuous, convex, and monotonically decreasing functions, $u_N(z)$ has the same properties. When (A3) is satisfied by all agents, $u_N(z)$ also has a unique zero-crossing, which we denote as $\zc_N$. 
Define the frontier of the sub-economy without agent $i$ as $u_{N\backslash \{i\}}(z) \triangleq \max_{j\neq i}u_j(z)$, and  the \emph{$m^{\text{th}}$ frontier} of the economy as the $m^{\mathrm{th}}$ upper envelope of $\{u_i(z)\}_{i\in N}$. See Figure~\ref{fig:P1P5_characterization}.

\begin{figure}[t!]
\centering   
\begin{tikzpicture}[scale = 1][font=\small]


\draw[-, name path = uA] (-0.2, 3.3) to[out=-65, in = 150] (2.5, -0.2);

\draw[dashed, name path = uB] (-0.2, 2.2) to[out=-35, in = 175] (6,-0.1);

\draw[dashdotted,name path = uC] (-0.2, 1.2) parabola[bend at end] (5,-0.1);

\draw[dotted] (-0.2, 0.5) parabola[bend at end] (1.5,-0.1);

\path [name intersections={of=uA and uB, by=intAB}];
\path [name intersections={of=uA and uC, by=intAC}];

\draw[blue!80, line width= 1mm, opacity=0.5] (-0.2, 3.3)  to[out=-65, in=-240] (intAB);
\draw[blue!80, line width=1mm, opacity=0.5] (intAB) to[out=-32, in = 175] (6,-0.1);

\draw[yellow!80, line width= 1mm, opacity=0.5] (-0.2, 2.2)  to[out=-35, in=-213] (intAB);

\draw[yellow!80, line width=1mm, opacity=0.5] (intAB) to[out=-60, in = -220] (intAC) parabola[bend at end] (5,-0.1);


\fill [pattern = north east lines, pattern color = black!40, opacity=0.7] (-0.2, 2.2)--(-0.2, 3.3)--(intAB)--(-0.2, 2.2);

\fill [pattern = north east lines, pattern color = black!40, opacity=0.7] (intAB) to[out=-32, in = 172] (5.2, 0) --(3.55, 0) --(intAC)--(intAB);

\draw[->] (-0.3,0) -- (6.5,0) node[anchor=north] {$z$};

\draw[->] (0,-0.3) -- (0, 3.5) node[anchor=west] {$y$};

\draw (5.3, 0) node[anchor = north]{{\scriptsize $\zc_N$}};

\filldraw [black] (5.2, 0) circle (1pt);
\draw (5.3, 0) node[anchor = south]{{\scriptsize $A$}};

\draw [blue!80, line width=1mm, opacity=0.5] (7, 3.4) -- (7.5, 3.4);
\draw (7.5, 3.4) node[anchor=west] {\footnotesize{The frontier}};

\draw [yellow!80, line width=1mm, opacity=0.5] (7, 2.9) -- (7.5, 2.9);
\draw (7.5, 2.9) node[anchor=west] {\footnotesize{The 2nd frontier}};

\fill [pattern = north east lines, pattern color = black!40, opacity=0.5](7, 2.33) -- (7.5, 2.33) -- (7.5, 2.53) -- (7, 2.53);
\draw	(7.5, 2.4) node[anchor=west] { {\footnotesize Possible payments}};

\draw [-] (7, 1.9) -- (7.5, 1.9) node[anchor=west] {\footnotesize{$u_1(z)-y = 0$}};

\draw [dashed] (7, 1.4) -- (7.5, 1.4) node[anchor=west] {\footnotesize{$u_2(z) - y = 0$}};

\draw [dashdotted] (7, 0.9) -- (7.5, 0.9) node[anchor=west] {\footnotesize{$u_3(z) - y = 0$}};

\draw [dotted] (7, 0.4) -- (7.5, 0.4) node[anchor=west] {\footnotesize{$u_4(z) - y = 0$}};

\end{tikzpicture}
\caption{Characterization of possible outcomes for mechanisms under (P1)-(P6).	
\label{fig:P1P5_characterization}} \vsq{-0.5em}
\end{figure}
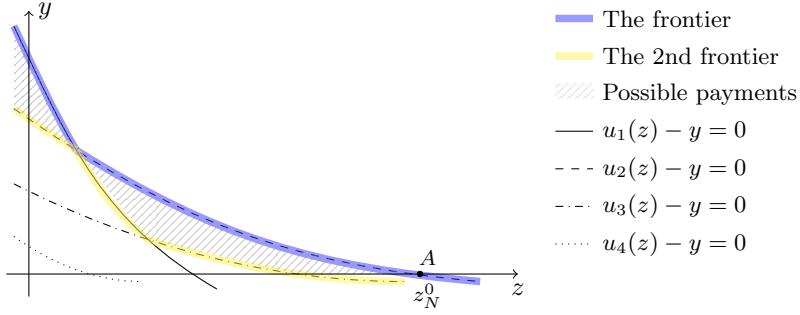

We first characterize the possible outcomes for any two-period mechanism satisfying (P1)-(P6) in the following two lemmas. This characterization is crucial for our main technical results. 


\begin{restatable} 
{lemma}{lemCharacterization}\label{lem:lem_P1P5_characterization} Assume that the type space includes all value distributions satisfying (A1)-(A3), and consider a two-period mechanism that satisfies (P1)-(P6). For any type profile $\CDF$, 
the allocated agent $\winner$ and the two-part payment $(z^\ast, y^\ast)$ agent $\winner$ is charged satisfy:
\begin{enumerate}[(i)]
	\setlength\itemsep{0em}
	\item $(z^\ast, y^\ast)$ resides weakly below $u_{\winner}(z)$.
%
	\item $(z^\ast, y^\ast)$ resides weakly above the frontier of the rest of the economy $u_{N \backslash \{\winner\}}(z)$. 
	\item The allocated agent faces a non-negative base payment $y^\ast \geq 0$.
\end{enumerate}
\end{restatable}

Instead of requiring the type space to include all value distributions satisfying (A1)-(A3), the lemma also holds assuming that the type space is the set of all $(\fixedV_i,\fixedP_i)$ types. 
We defer the full proof to Appendix~\ref{appx:proof_P1P5}, giving intuition here.
Part (i) is implied by IR. 
If (ii) is violated, i.e. there exists agent $i \neq \winner$ s.t. $u_i(z^\ast)- y^\ast > 0$, then in the economy where the type of agent $\winner$ is also given by $\CDF_{i}$, pretending that her type is $\CDF_{\winner}$ is a useful deviation. 
We show this by proving 
that any agent who is tied with some other agent cannot get strictly positive utility. 
Part (iii) is proved by showing that if the allocated agent is charged $y^\ast< 0$ in some economy, we can replace the agent's type with some $(\fixedV_i,\fixedP_i)$ type, in which case either IC or ND is violated.

\medskip

Given (P4) anonymity, regardless of whether there are ties in agents' reports, there is a two-part payment $(z^\ast, y^\ast)$ a mechanism charges its allocated agent(s). 
Lemma~\ref{lem:lem_P1P5_characterization} implies that $(z^\ast, y^\ast)$ is in between the first and second frontiers and above the horizontal axis $y = 0$ (see Figure~\ref{fig:P1P5_characterization}), and that the allocated agent(s) resides on the frontier at $z^\ast$ (i.e. $u_N(z^\ast) = u_\winner(z^\ast)$ if agent $\winner$ is allocated). The following lemma proves monotonicity properties w.r.t. the penalty $z^\ast$, on utilization and welfare achieved by mechanisms that satisfy (P1)-(P6), for any fixed economy.


\begin{restatable}{lemma}{lemUtSwMon} \label{lem:range_of_ut_sw} 
Fix any type profile $\CDF$ satisfying (A1) and (A2). Among all mechanisms that satisfy (P1)-(P6), the utilization achieved by a mechanism is (weakly) higher if it charges its allocated agent(s) a higher penalty $z^\ast$. Similarly, the achieved welfare is monotonically increasing in $z^\ast$ when $z^\ast \leq \socialV$, and monotonically decreasing in $z^\ast$ when $z^\ast > \socialV$.  
\end{restatable}

The proof is provided in Appendix~\ref{appx:proof_lem_range_of_ut_sw}, which uses the monotonicity properties in Lemmas~\ref{lem:util_welfare} and~\ref{lem:crossing_utilities} and the characterization in Lemma~\ref{lem:lem_P1P5_characterization}. An important implication of this lemma is that the highest possible welfare achievable by any mechanism under (P1)-(P6) is achieved by charging a penalty $z^\ast = \socialV$ if $\max_{i \in N} u_i(\socialV) \geq 0$ and $z^\ast = \max_{i \in N} \zc_i$ otherwise. Lemma~\ref{lem:lem_P1P5_characterization} then requires allocating to agent(s) in $\arg \max_{i \in N} u_i(z^\ast) $, which is in fact the set of agents allocated under CP($\socialV$).
Therefore, the only ways to achieve an even higher welfare than the CP($\socialV$) mechanism are (i) break ties in favor of higher welfare instead of at random, and (ii) charge a higher penalty, when the CP($\socialV$) penalty determined by the the second highest bid is lower than optimal $z^\ast$. 


\begin{definition}[Generic input] A type profile $\CDF$ satisfies the \emph{generic input} property if for any $i, j \in N$, $i \neq j$: (i) $u_i(\socialV) \neq u_j(\socialV)$, if $u_i(\socialV),~ u_j(\socialV) \geq 0$, and (ii) $\zc_i \neq \zc_j$, if $u_i(\socialV),~ u_j(\socialV) < 0$.
\end{definition}

A type profile is generic if no two agents have the same period~$0$ willingness to pay given penalty $\socialV$, or the same maximum acceptable penalty that is below $\socialV$. As a result, there would not be any tie under the CP($\socialV$) mechanism.\footnote{ 
  The generic inputs assumption is only needed for the indirect, CP($\socialV$) mechanism. 
A direct revelation version, that always breaks ties in favor of the agent with higher utilization, has all the performance guarantees stated in Corollary~\ref{cor:cpm_opt_wipi} and Theorems~\ref{thm:cpm_not_dom} and~\ref{thm:cpm_opt} without the generic input assumption.}
An immediate result is that the CP mechanism is welfare optimal for the $(\fixedV_i, \fixedP_i)$ type space with the generic input assumption. This is easy to see, since in this type domain, a higher penalty does not improve utilization, or induce more welfare-optimal time $1$ utilization decision of the allocated agent.
\begin{corollary} \label{cor:cpm_opt_wipi} Assume the type space is the set of all $(\fixedV_i, \fixedP_i)$ value distributions. With the generic input assumption, the CP($\socialV$) mechanism is welfare-optimal type profile by type profile among all two-period mechanisms that satisfy (P1)-(P6). 
\end{corollary}

We also have the following result, the first of our two main results. Theorem~\ref{thm:cpm_not_dom} states that the CP($\socialV$) mechanism is not dominated in welfare by any two-period mechanism under (P1)-(P6).
\begin{restatable}{theorem}{thmNotDominated} \label{thm:cpm_not_dom} 
Assume the type space is the set of all value distributions satisfying (A1) and (A2). With the generic input assumption, no two-period mechanism under (P1)-(P6) achieves weakly higher social welfare than the CP($\socialV$) mechanism for all type profiles, and a strictly higher social welfare than the CP($\socialV$) mechanism for at least one type profile.
\end{restatable}

See Appendix~\ref{appx:proof_opt_cpm} for the proof. Intuitively, if a mechanism $\mech$ under (P1)-(P6) always achieves weakly higher welfare than the CP($\socialV$) mechanism, lemmas~\ref{lem:lem_P1P5_characterization} and ~\ref{lem:range_of_ut_sw} require that it always allocates the resource to the winner under CP($\socialV$). We then show a violation of either IR or DSE, if $\mech$ ever charges the a higher penalty than the CP($\socialV$) mechanism does to improve welfare. 

A payment space $\paymentSpace$ of a mechanism is the set of two-part payments that's achievable by some report profile of the agents:  $\paymentSpace =  \{ (\tone_\winner(\report), \tzero_\winner(\report)) ~|~ \report \in \reportSet^n \}$.
\begin{definition}[Ordered payment space] A payment space $\paymentSpace \subseteq \setR^2$ is \emph{ordered} if all agents with types satisfying (A1) and (A2) agree on which one of any two pairs of payments is more preferable. Formally, $\forall (z, y) \in \paymentSpace$, $\forall (\tilde{z}, \tilde{y}) \in \paymentSpace$, for all $F_1$, $F_2$ under (A1) and (A2),  
\begin{align*}
	u_1(z) - y > u_1(\tilde{z}) - \tilde{y} \Rightarrow u_2(z) - y \geq u_2(\tilde{z}) - \tilde{y}.
\end{align*}
\end{definition}



The second main result is that the CP($\socialV$) mechanism is
welfare-optimal profile by profile among a large class of mechanisms that always allocate the resource, and use an ordered payment space. 

\begin{restatable}{theorem}{thmOptOrderedSpace} \label{thm:cpm_opt} Assume the type space is the set of all value distributions satisfying (A1) and (A2). With the generic input assumption, the CP($\socialV$) mechanism is welfare-optimal  type profile by type profile, among all two-period mechanisms that satisfy (P1)-(P6), always allocate the resource, and use an ordered payment space.
\end{restatable}

We defer the full proof to Appendix~\ref{appx:proof_opt_cpm}. Intuitively, if a mechanism achieves a higher welfare by charging a penalty larger than the penalty determined the CP$(\socialV)$ mechanism, we may construct an alternative economy, and show that guaranteeing (P1)-(P6) results in a violation of the assumption that the payment space is ordered. 
%

\medskip

In fact, all mechanisms discussed so far use ordered payment spaces. The second price auction always charges no penalty and a non-negative base payment, thus $\paymentSpace_{\txtSP} = \{ (z,y) \in \setR^2 ~|~ z=0,~y\geq 0\}$, as illustrated in Figure~\ref{fig:pspace_sp}. Similarly, $\paymentSpace_{\txtCSP} = \{ (z,y) \in \setR^2 ~|~ z\geq 0,~y = 0\}$, as in Figure~\ref{fig:pspace_csp}. The CP($\maxZ$) mechanism sets payments $\paymentSpace_{\mathrm{CP}(\maxZ)} = \{ (z,y) \in \setR^2 ~|~ 0\leq z \leq \maxZ,~y = 0\} \cup \{ (z,y) \in \setR^2 ~|~ z = \maxZ, ~ y \geq 0\}$, as illustrated in Figure~\ref{fig:pspace_cpm}.
We may consider other mechanisms, for example, a \emph{SP+$\maxZ$ mechanism} collects a single non-negative bid from each agent, allocates to the highest bidder (if there is any), charges a time $1$ penalty $\maxZ$, and the second highest bid as the base payment (no base payment if only one agent submitted non-negative bid). The payment space is given in Figure~\ref{fig:pspace_spz}. We may also consider a \emph{$\gamma$-CSP mechanism}, which collects a single bid from each agent, allocates to the highest bidder, charges a $\gamma$ fraction of the second highest bid as the time $0$ base payment, and the rest of it as the no-show penalty. The payment space is as shown in Figure~\ref{fig:pspace_gammaCSP}.


\newcommand{\paymentSpaceScale}{1.3}
\newcommand{\figSpace}{1.2}
\begin{figure}
\centering 
\subfloat[\small{$\paymentSpace_{\txtSP}$}]{\label{fig:pspace_sp}
\begin{tikzpicture}[scale = \paymentSpaceScale][font=\small]
\draw[->] 	(-0.5,0) -- (1,0) node[anchor=north] {$z$};
\draw[->] 	(0,-0.5) -- (0,1) node[anchor=west] {$y$};
\draw[-, very thick] (0,0) -- (0,1);
\end{tikzpicture}}
\hspace{\figSpace em}
\subfloat[\small{$\paymentSpace_{\txtCSP}$}]{\label{fig:pspace_csp}
\begin{tikzpicture}[scale = \paymentSpaceScale][font=\small]
\draw[->] 	(-0.5,0) -- (1,0) node[anchor=north] {$z$};
\draw[->] 	(0,-0.5) -- (0,1) node[anchor=west] {$y$};
\draw[-, very thick] (0,0) -- (1,0);
\end{tikzpicture}}
\hspace{\figSpace em}
\subfloat[\small{$\paymentSpace_{\txtCPM(\maxZ)}$}]{\label{fig:pspace_cpm}
\begin{tikzpicture}[scale = \paymentSpaceScale][font=\small]
\draw[->] 	(-0.5,0) -- (1.2,0) node[anchor=north] {$z$};
\draw[->] 	(0,-0.5) -- (0,1) node[anchor=west] {$y$};
\draw[-, very thick] (0,0) -- (0.7,0) -- (0.7, 1);
\draw (0.7, -0) node[anchor=north] {\footnotesize{$\maxZ$}};
\end{tikzpicture}}
\hspace{\figSpace em}
\subfloat[\small{$\paymentSpace_{\mathrm{SP+\maxZ}}$}]{\label{fig:pspace_spz}
\begin{tikzpicture}[scale = \paymentSpaceScale][font=\small]

\draw[->] 	(-0.5,0) -- (1,0) node[anchor=north] {$z$};
\draw[->] 	(0,-0.5) -- (0,1) node[anchor=west] {$y$};

\draw[-, very thick] (0.7, 0) -- (0.7, 1);
\draw (0.7, -0) node[anchor=north] {\footnotesize{$\maxZ$}};
\end{tikzpicture}}
\hspace{\figSpace em}
\subfloat[\small{$\paymentSpace_{\mathrm{\gamma-CSP}}$}]{\label{fig:pspace_gammaCSP}
\begin{tikzpicture}[scale = \paymentSpaceScale][font=\small]
\draw[->] 	(-0.5,0) -- (1,0) node[anchor=north] {$z$};
\draw[->] 	(0,-0.5) -- (0,1) node[anchor=west] {$y$};
\draw[-, very thick] (-0, -0) -- (1,0.8);
\end{tikzpicture}}
\caption{Ordered payment spaces for various two-period mechanisms. 	
\label{fig:ordered_payment_spaces}} 
\vsq{-0.5em}
\end{figure}
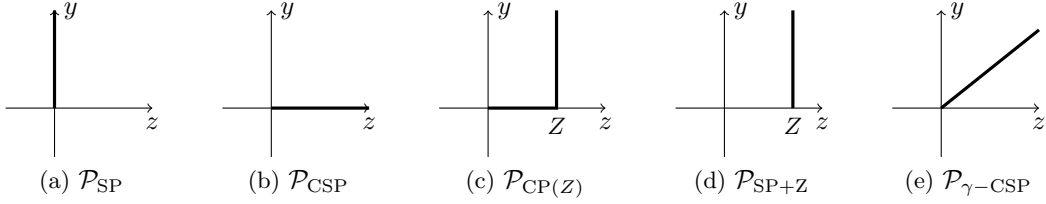

In addition, mechanisms with an ordered payment space support a simple indirect message structure: the interpretation is that the mechanism asks an agent to report the largest payment in the ordering that is acceptable. Given this, dominant-strategy mechanisms can be achieved by allocating to the agent with the highest report and charging the second-highest report (defined with respect to the payment order).

\subsection{Uniqueness and Optimality of CSP} \label{sec:uniq_opt_CSP}

The CSP mechanism can be considered as a special case of the CP$(\socialV)$ mechanism with no upper bound on the penalty an agent may be charged.  
Recognizing this, and defining \emph{generic input} as no ties in agents' zero-crossings, we can obtain the following optimality results for utilization, analogous to the optimality results of CP$(\socialV$) presented in the previous section. Under the additional assumption of no-charge when using the resource, we can state a uniqueness result for CSP.
\begin{restatable}{theorem}{thmUniqOptCSP} \label{thm:csp_uniq_opt} Assume the type space is the set of all value distributions satisfying (A1)-(A3), assume generic input, and consider only two-period mechanisms that satisfy (P1)-(P6): 
\begin{enumerate}[(i)]
	\setlength\itemsep{0em}
	\item the CSP mechanism is the unique mechanism that always allocates the resource, and does not charge the allocated agent if the resource is utilized (i.e. ``no-charge").
	\item for the $(\fixedV_i, \fixedP_i)$ type space, the CSP mechanism is optimal for utilization, profile by profile.
	\item the CSP mechanism is not dominated for utilization by any mechanism.
	\item the CSP mechanism is utilization optimal profile by profile, among all mechanisms that always allocate the resource  and use an ordered payment space.
\end{enumerate}
\end{restatable}

The proofs are similar to the proofs of the optimality results in the previous section, observing that the only way to achieve a higher utilization than CSP is to allocate the resource to the CSP winner, and charge the CSP winner a strictly higher penalty than the second highest bid. See Appendix~\ref{appx:proof_opt_csp} for a proof of this theorem.

\section{Assignment of Multiple Resources} \label{sec:multi_resource}

In this section, we generalize the model and mechanisms for assigning multiple resources, but where each agent remains interested in receiving at most one resource (i.e., the unit-demand model).
Let $N = \{1, 2, \dots, n\}$ be the set of agents and $M = \{a, b, \dots, m\}$ be the set of $m$ resources. For each $a \in M$, the value for each agent $i$ to use resource $a$ is a random variable $V_{i,a}$ with CDF $\CDF_{i,a}$. $\{\CDF_{i,a}\}_{a \in M}$ corresponds to agent $i$'s type, and we assume that the random values 
are independent. 

The SP auction can be generalized as the VCG mechanism~\cite{vickrey1961counterspeculation,clarke1971,groves1973}, where it is a dominant strategy for agent $i$ to bid $u_{i,a}(0) = \mathbb{E}[V_{i,a}^+]$ for resource $a$.
The naive generalization of the CP$(Z)$ mechanism (which assigns the resources to maximize the sum of the two-part bids, and charges each agent the externality that she imposes on the rest of the economy in terms of the two-part bids) fails to be incentive compatible. This is because agents' expected utilities are not quasi-linear in the period 1 penalty payments.

A set of two-part payments $\{(z_a,y_a)\}_{a\in M}$ is a set of competitive equilibrium (CE) price if when each agent selects her favorite resource given these payments, no resource is selected more than once, and a resource that is not selected has zero prices $z_a = y_a = 0$. 
Recall that in the message space of a CP$(Z)$ mechanism, a two-part payment $(z,y)$ is ``higher" if it has a larger sum $z+y$, and that an agent has a lower expected utility if she is charged a higher two-part bid.
We generalize the CP$(Z)$ mechanism as the minimum CE price mechanism~\cite{DBLP:journals/ior/AlaeiJM16,demange1985strategy}:

\begin{definition}[Generalized CP($Z$) mechanism] The generalized CP mechanism parametrized by maximal penalty $Z$ (the GCP($Z$) mechanism) collects value distributions $\{\CDF_{i,a}\}_{i \in N, a \in M}$ from the agents, and computes the minimum CE payments $\{(z_a, y_a) \}_{a \in M}$ in the payment space $\reportSet = \{(z,y) \in \setR^2 ~|~0 \leq z \leq Z,~y = 0 \} \cup \{(z,y) \in \setR^2 ~|~ z = Z, ~ y \geq 0 \}$. 
\begin{enumerate}[$\bullet$]
	\item Allocation rule: for each agent $i \in N$, if $\max_{a \in M} u_{i,a}(z_a) - y_a \geq 0$, then $x_i(F) = a_i^\ast \in \arg \max_{a \in M} u_{i,a}(z_a) - y_a$ (breaking ties to clear the market). 
	\item Payment rule: charge each agent $\tone_i(F) = z_{a_i^\ast}$ and $\tzero_i(F) = y_{a_i^\ast}$ if agent $i$ is allocated resource $a_i^\ast$. All other payments are zero. 
\end{enumerate}
\end{definition}

For the case where the $m$ resources are identical, the mechanism reduces to the $(m+1)\th$ price version of the CP($Z$) mechanism.
Intuitively, each agent is assigned one of her favorite resources given the prices $\{(z_a, y_a)\}_{a \in M}$, if she can afford any.
When $Z = 0$, the mechanism reduces to VCG, and for the case when $Z = \infty$, we get the generalized CSP (GCSP) mechanism, which prices each resource at the minimum CE penalties. Demange and Gale~\cite{demange1985strategy} prove that the minimum CE price mechanism is incentive compatible, and Alaei et al.~\cite{DBLP:journals/ior/AlaeiJM16} provide a recursive algorithm to compute these minimum CE prices. 

\begin{theorem} Given assumptions (A1)-(A2), under the generalized CP($Z$) and the generalized CSP mechanisms, it is a dominant strategy for each agent to truthfully report her type.
\end{theorem}


Unlike the assignment of a single resource, the welfare and utilization under the generalized CP mechanisms need not dominate that of the  VCG mechanism. Still, simulation results in Section~\ref{sec:simulations} show that the generalized CP mechanism achieves significantly higher average welfare and utilization than VCG.

\section{Simulation Results} \label{sec:simulations}

In this section, we compare the welfare and utilization achieved by different mechanisms for assigning a single resource, or multiple heterogeneous resources.

We adopt the natural exponential type model (see Example~\ref{ex:exp_model}), under which agent $i$'s value for using resource $a$ is $V_{i,a} = w_{i,a} - O_{i,a}$, where $w_{i,a}>0$ is the fixed value from using the resource, and $O_{i,a} \sim \mathrm{Exp}(\lambda_{i,a})$ is the exponentially distributed opportunity cost. $\E{V_{i,a}} = w_{i,a} - \lambda_{i,a}^{-1}$ where $\lambda_{i,a}^{-1}$ is the expected value of the opportunity cost. 
We consider the type distribution where the value and the expected opportunity cost $\lambda_{i,a}^{-1}$ are uniformly distributed: $\lambda_{i,a}^{-1} \sim \mathrm{U}[0,~L]$ and $w_{i,a} \sim \mathrm{U}[0,~\lambda_{i,a}^{-1}]$. With $w_{i,a} < \lambda_{i,a}^{-1}$, (A1)-(A3) are satisfied.

\subsection{Single Resource Assignment}

We first study the assignment of a single resource. We set $L = 10$ and $\socialV = 5$, corresponding to scenario where the societal value $\socialV$ is equal to the expected opportunity cost for an average agent to use a resource.\footnote{ Appendix~\ref{appx:additional_simulations} presents additional results for settings where the societal preference for utilization is weaker/stronger. 
}
Varying the number of agents from 2 to 15, we compute the average social welfare and utilization over 10,000 randomly generated profiles under the CP$(\socialV)$, CSP, SP mechanisms, and other benchmarks. 
See Figure~\ref{fig:W5}.  
The {\em First-Best} benchmark is the highest achievable welfare and utilization, subject to the assumptions of IR and \emph{no deficit} (ND), i.e. the expected total revenue of the mechanism has to be non-negative. 
%
%
The {\em Random} benchmark assigns the resource at random to one of the agents without charging any payment, modeling the first-come-first-serve system of reserving the resource.


\newcommand{\figScale}{0.85}

\begin{figure}[t!]
\centering
\subfloat[Social Welfare]{\label{fig:welfare_W5} 
	\includegraphics[scale=\figScale]{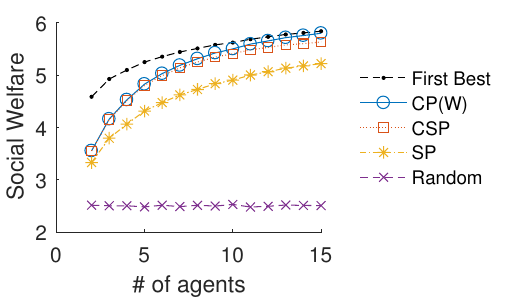}
}
%
\subfloat[Utilization]{\label{fig:utilization_W5}
 	\includegraphics[scale=\figScale]{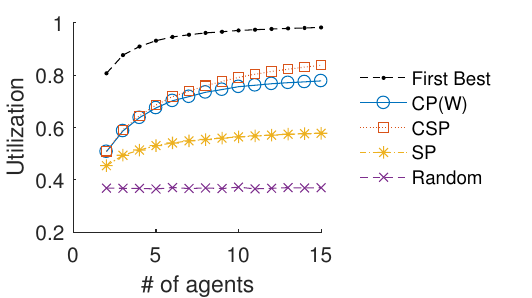}
}
\caption{Social welfare and utilization for a single resource. 
\label{fig:W5} 
}
\vsq{-1em}
\end{figure}

Figure~\ref{fig:welfare_W5} shows that 
the CP($\socialV$) mechanism achieves slightly higher welfare than the CSP mechanism, and is very close to the first-best welfare assuming IR and ND. Both CP($\socialV$) and CSP achieve better social welfare than the SP auction.
The average utilization under the mechanisms are shown in Figure~\ref{fig:utilization_W5}. 
In comparison to the SP auction, both CSP and CP$(\socialV)$ mechanisms achieve significantly higher utilization.%
\footnote{
As the number of agents increases, the CP$(\socialV)$ mechanism approaches the first-best welfare, whereas it is curious that there remains a gap to the first-best utilization for CSP. With sufficient competition, the  under CP$(\socialV)$ is the agent that achieves the first-best welfare, and moreover, CP$(\socialV)$ determines the optimal penalty $\socialV$. In contrast, the CSP winner may not be the one that achieves the first-best utilization, and the second-highest bid remains lower than the optimal (first-best) penalty.
}

\subsection{Assigning Multiple Heterogeneous Resources}

\begin{figure}[t!]
\centering
\subfloat[Social Welfare]{\label{fig:welfare_3items_W5} 
	\includegraphics[scale=\figScale]{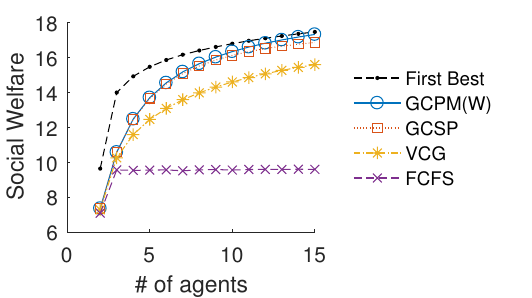}
}
%
\subfloat[Utilization]{\label{fig:utilization_3items_W5}
 	\includegraphics[scale=\figScale]{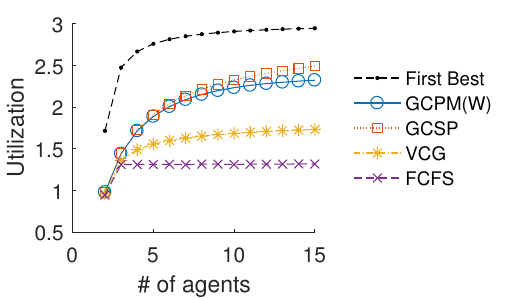}
}
\caption{Welfare and utilization for 3 heterogeneous resources, with small societal value $\socialV = 5$.
\label{fig:W5_3items} 
}
\vsq{-1em}
\end{figure}

1
We now compare the social welfare and utilization (expected number of utilized resources) for assigning $3$ heterogeneous resources, as the number of agents varies from $2$ to $15$.  Figure~\ref{fig:welfare_3items_W5} shows the average welfare and utilization achieved by different mechanisms and benchmarks, for the scenario where $\socialV = 5$ is equal to the expected opportunity cost for an agent to use a resource. 
The First Come First Serve (FCFS) benchmark allows each agent to choose her favorite remaining resource as they arrive in a random order, and does not charge any payments. The generalized CP$(\socialV)$ and generalized CSP mechanisms achieve better welfare and significantly better utilization than the VCG mechanism.

\begin{figure}[t!]
\centering
\subfloat[Social Welfare]{\label{fig:welfare_3items_W0} 
	\includegraphics[scale=\figScale]{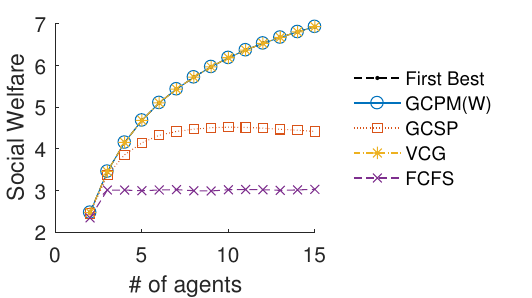}
}
%
\subfloat[Utilization]{\label{fig:utilization_3items_W0}
 	\includegraphics[scale=\figScale]{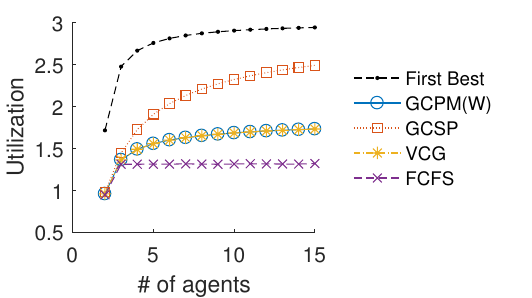}
}
\caption{Welfare and utilization for 3 heterogeneous resources, with small societal value $\socialV = 0$.
\label{fig:W0_3items} 
}
\end{figure}

We also present results when the society has weaker or stronger preference for utilization.
Figure~\ref{fig:W0_3items} shows the average welfare and utilization, 
when the societal value from utilization is $\socialV = 0$. Both the generalized CP$(0)$ mechanism and the VCG mechanism achieve the first-best welfare, whereas the generalized CSP mechanism is less efficient but achieves a higher utilization. 

\begin{figure}[t!]
\centering
\subfloat[Social Welfare]{\label{fig:welfare_3items_W10} 
	\includegraphics[scale=\figScale]{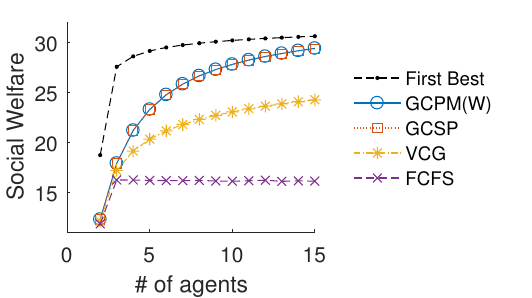}
}
%
\subfloat[Utilization]{\label{fig:utilization_3items_W10}
 	\includegraphics[scale=\figScale]{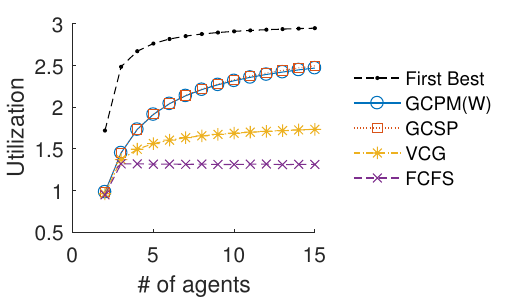}
}
\caption{Welfare and utilization for 3 heterogeneous resources, with small societal value $\socialV = 10$. \label{fig:W10_3items}  \vspace{-0.5em}
}
\end{figure}

Figure~\ref{fig:W10_3items} shows the average welfare and utilization for three resources, when the societal value from utilization is $\socialV = 10$. We can see that with a stronger desire for the resource being utilized, the generalized CP mechanism and the generalized CSP mechanism significantly outperform the VCG mechanism on both welfare and utilization. With a large number of agents, the achieved welfare is in fact close to the first-best benchmark under the IR and ND assumptions.

\section{Conclusion} \label{sec:conclusions}

We study the problem of resource assignment where agents have uncertainty about their values and where it is in the interest of the mechanism designer that resources be used and not wasted.
We introduce the contingent payment mechanisms parametrized by maximum penalty $\maxZ$, and show that among mechanisms with a set of natural axiomatic properties, the CP($\socialV$) mechanism is not dominated, and is welfare optimal among mechanisms that always allocate the resource and support an especially simple indirect structure. 
We prove similar optimality results for the contingent second price mechanism for utilization, extend the results to the assignment of multiple, identical items, and show that the mechanisms can also be generalized to assign multiple heterogeneous resources. 
Simulation results demonstrate the effectiveness and robustness of our mechanisms.

An interesting direction for future work is to generalize the model to allow for more than two time periods, where agents may arrive asynchronously, uncertainty unfolds gradually over time, and resources can be re-allocated. It would also be interesting to fold in considerations from behavioral economics, understanding how these contingent mechanisms interact with present-biased agents.  

\bibliography{CSP_refs}

\newpage

\appendix

\noindent{}\textbf{\huge{Appendix}}

\vspace{1.5em}

\noindent{}%
Derivation of expected utility functions and DSE bids for various type models are provided in Appendix~\ref{appx:bids}.
We provide in Appendix~\ref{appx:proofs} the proofs that are omitted from the body of the paper. %
Additional simulation results and discussions are presented in 
Appendix~\ref{appx:additional_simulations}.

\medskip

\section{Utilities and DSE Bids under Different Type Models} \label{appx:bids}

\subsection{The $(\fixedV_i, \fixedP_i)$ Model}

Recall that for the $(\fixedV_i, \fixedP_i)$ type model introduced in Example~\ref{ex:vipi}, the distribution of an agent's random value is given by
$$
	V_i = \pwfun{\fixedV_i & \txtwp \fixedP_i \\ -\infty & \txtwp 1 - \fixedP_i}	
$$ 
for some $\fixedV_i > 0$ and $\fixedP_i \in (0,1)$. The CDF of $V_i$ is
$$
	F_i(v) = \pwfun{1 - \fixedP_i & \txtfor v < \fixedV_i \\ 1 & \txtfor v \geq \fixedV_i }.
$$	
The expected utility of an agent, when allocated and charged a two-part payment $(z,y)$, is 
$$
	u_i(z) - y = \pwfun{\fixedV_i \fixedP_i  - (1 - \fixedP_i) z - y, & \txtfor z \geq - \fixedV_i \\
		-z -y, & \txtfor z < - \fixedV_i}.
$$
This is because when $z$ is negative, the agent uses the resource if she is able to only if $\fixedV_i \geq -z$. If $\fixedV_i < -z$, the agent would decides to never use the resource, and always pays $-z$ in period~1.
The DSE bids under the different mechanisms are:
\begin{enumerate}[$\bullet$]
	\item For the SP auction: $b_{i,\mathrm{SP}}^\ast = u_i(0) =  \fixedV_i \fixedP_i$ 
	\item For the CSP mechanism: $b_{i,\mathrm{CSP}}^\ast = \zc_i = \fixedV_i \fixedP_i / (1 - \fixedP_i)$.
	\item For the CP$(Z)$ mechanism: $b_{i, \mathrm{CP}}^\ast = (\bone_i, \bzero_i)$, where 
\begin{itemize}
	\item if $\zc_i \geq Z$: $\bone_i = Z$ and $\bzero_i = u(Z) = \fixedV_i p_i - Z(1- \fixedP_i)$, and
	\item if $\zc_i < Z$: $\bone_i = \zc_i = \fixedV_i \fixedP_i / (1-\fixedP_i)$ and $\bzero_i = 0$.
\end{itemize}	
	
\end{enumerate}

When an agent with $(\fixedV_i,\fixedP_i)$ type is assigned the resource and charged a penalty $z \geq -\fixedV_i$, the utilization is $ut_i(z) = \fixedP_i$, and the welfare is $\fixedP_i(\fixedV_i + \socialV)$. This is the first-best welfare and utilization that is achievable from allocating the resource to a $(\fixedV_i, \fixedP_i)$ type agent, since the penalty does not incentivize the agent to use the resource more or less often.

\subsection{The Exponential Model} Consider the exponential model introduced in Example~\ref{ex:exp_model}, where the random value for an agent to use the resource is equal to a fixed value $w_i > 0$ minus an opportunity cost $O_i$ that is exponentially distributed  with parameter $\lambda_i$:
\begin{align*}
	V_i = \fixedV_i - O_i, \text{ where } O_i \sim \mathrm{Exp}(\lambda_i)
\end{align*}
The CDF and PDF of the random value are
\begin{align*}
	F_i(v) = \pwfun{e^{\lambda_i(v - w_i)}, & \txtfor v \leq w_i \\ 1, &\txtfor v > w_i},
	~f_i(v) = \pwfun{\lambda_i e^{\lambda_i (v - w_i)}, & \txtfor v \leq w_i \\ 0, & \txtfor v > w_i}.
\end{align*}
Note that the highest value an agent may get from using the resource is $w_i$. When charged a two part payment $(z,y)$ with $z < -w_i$, the agent is \emph{paid} a positive amount that is larger than $w_i$ for not using the resource, thus the agent never uses the resource, and gets expected utility $u_i(z) = -z - y$. When $z \geq -w_i$, the agent's expected utility as a function of the two-part payments is
\begin{align*}
	u_i(z) - y = \int_{-z}^{w_i} v \lambda_i e^{\lambda_i(v - w_i)}  dv - z e^{-\lambda_i(z + w_i)} - y = w_i  + \frac{1}{\lambda_i} \left( e^{-\lambda_i(w_i + z)} - 1 \right) - y.
\end{align*}
The DSE bids under the different mechanisms are:
\begin{enumerate}[$\bullet$]
	\item For the SP auction: $b_{i,\mathrm{SP}}^\ast = u_i(0) = w_i  + \frac{1}{\lambda_i} \left( e^{-\lambda_i w_i} - 1 \right)$.
	\item For the CSP mechanism: $b_{i,\mathrm{CSP}}^\ast = \zc_i = - w_i - \frac{1}{\lambda_i} \log(1 - w_i \lambda_i) $.
	\item For the CP$(Z)$ mechanism: $b_{i, \mathrm{CP}}^\ast = (\bone_i, \bzero_i)$, where 
\begin{itemize}
	\item if $\zc_i \geq Z$: $\bone_i = Z$ and $\bzero_i = u_i(Z) = w_i  + \frac{1}{\lambda_i} \left( e^{-\lambda_i(w_i + Z)} - 1 \right)$, and
	\item if $\zc_i < Z$: $\bone_i = \zc_i = - w_i - \frac{1}{\lambda_i} \log(1 - w_i \lambda_i)$ and $\bzero_i = 0$.
\end{itemize}
\end{enumerate}
%
%
%
%
An agent's utilization as the function of the no-show penalty is of the form:
\begin{align*}
	ut_i(z) = \pwfun{0, &\txtfor z < - w_i \\
		1 - e^{-\lambda_i(w_i + z)}, & \txtfor z \geq - w_i },
\end{align*}
and the expected revenue of the mechanism is $rev_i(z,y) = z \cdot e^{-\lambda_i(w_i + z)} + y.$
Setting $u_i(z) = rev_i(z,0)$ (i.e. solving for the crossing point of the zero-profit curve and the zero-revenue curve, as illustrated in Figure~\ref{fig:irbb}(ii)), we get the highest possible penalty that a mechanism can charge an agent with exponentially type as:
\begin{align*}
	z^{FB}_i = -\frac{1}{\lambda_i} - \frac{1}{\lambda_i} \Omega \left(k, e^{-1 + w_i \lambda_i}(-1 + w_i \lambda_i) \right).
\end{align*}
Here, $\Omega(\cdot)$ is the omega function (which is also called the Lambert W function), the inverse function of $y = x e^x$. The $\Omega$ function is multivalued. We take the lower branch with $k = -1$, since when taking $k = 0$, $ \Omega \left(k, e^{-1 + w_i \lambda_i}(-1 + w_i \lambda_i) \right) = -1 + w_i\lambda_i \Rightarrow z_i^{FB} = -w$, which is not the highest possible penalty that we are looking for. Plugging $z_i^{FB}$ into the utilization function $ut_i(z)$, we can obtain the first-best utilization that is achievable by a mechanism that satisfies IR and No Deficit. 

When an agent with exponential type is charged a penalty $z \geq -w_i$, the achieved social welfare is $\sw_i(z) = w_i + (e^{-\lambda_i (w_i + z)} (1 + \lambda_i z)-1)/ \lambda_i + \socialV (1 - e^{-\lambda_i(z + w_i)})$. 
For the first best welfare, if $z^{FB}_i  \geq \socialV$, we can set 
$z = \socialV$, we achieve the first-best welfare of
\begin{align*}
	sw_i^{FB} = &  w_i + (e^{-\lambda_i (w_i + \socialV)} (1 + \lambda_i \socialV)-1)/ \lambda_i + \socialV (1 - e^{-\lambda_i(\socialV + w_i)}) \\
	= &  w_i + \socialV + (e^{-\lambda_i (w_i + \socialV)}-1)/ \lambda_i.
\end{align*}
When $z^{FB}_i < \socialV$, the first best welfare is simply $\sw_i(z_i^{FB})$.

\section{Proofs} \label{appx:proofs}

\subsection{Proof of Lemma~\ref{lem:util_welfare}} \label{appx:proof_lem_util_welfare}

\lemmaExpUtilityUtilization*

\begin{proof} We first prove part (i) of the lemma.

\medskip

\noindent{}\textit{Part (i).} As is outlined in the body of the paper, what is left to prove is $\Pm{V_i \geq -z} = 1 + u_i'(z+)$. Denote $V_i^- \triangleq - \min\{V_i, 0\}$, we know that $V_i = V_i^+ - V_i^-$, and that both $V_i^+$ and $V_i^-$ are non-negative random variables. First consider the case when $z \geq 0$, for which the expected utility function can be rewritten as  
\begin{align*}
	u_i(z) = \E{\max \{ V_i, -z \}} =\E{V_i^+ - \min\{V_i^-, z\} } = \E{V_i^+} - \E{\min\{V_i^-, z\}}.
\end{align*}
Observing $\min\{V_i^-, z\} = \int_0^z \one{V_i^- \geq v}dv$, and applying Fubini's theorem, we get
\begin{align*}
	\E{\min\{V_i^-, z\}} = \E{\int_0^z \one{V_i^- \geq v}dv} = \int_0^z \E{\one{V_i^- \geq v}} dv = \int_0^z \Pm{V_i^- \geq v} dv.
\end{align*}
Therefore,
\begin{align*}
	u_i(z)  = \E{V_i^+}  - \int_0^z \Pm{V_i^- \geq v} dv =  \E{V_i^+}  - \int_0^z \Pm{-V_i^- \leq -v} dv = \E{V_i^+} - \int_0^{z} F(- v) dv.
\end{align*}
From the fundamental theorem of calculus, the right derivative is equal to the right limit of the integrand, thus
\begin{align*}
	u_i'(z+) = - \lim_{v \rightarrow z+} F(-v) = - \lim_{v \rightarrow (-z)-}F(v) = - \Pm{V_i < -z} = \Pm{V_i \geq -z} - 1.
\end{align*}
%
We now consider the case where $z < 0$. The expected utility $u_i(z)$ can be rewritten as
\begin{align*}
	u_i(z) = &  \E{\max \{V_i, -z \}} 
	=  \int_0^\infty \Pm{\max \{V_i, -z \} \geq v} dv = -z + \int_{-z}^\infty \Pm{V_i \geq v}dv.
\end{align*}
Taking the right derivative, we again get $u_i'(z+) = -1 + \Pm{V_i \geq -z}$.

\medskip

\noindent{}\textit{Part (ii).} 
For any $z_1, z_2$ such that $ z_1 \leq z_2 \leq \socialV$, $-z_2 \leq V_i < -z_1 \Rightarrow V_i  +\socialV \geq 0$, thus we know 
\begin{align*}
	\sw_i(z_1) - \sw_i(z_2) =&  \E{(V_i +\socialV) \cdot \one{V_i \geq -z_1}} - \E{(V_i+\socialV) \cdot \one{V_i \geq -z_2}} \\
 	=& - \E{(V_i+\socialV) \cdot \one{-z_2 \leq V_i < -z_1}} \leq 0.
\end{align*}
Therefore $\sw(z_2) \geq \sw(z_1)$, and $\sw_i(z)$ is monotonically non-decreasing in this range. Similarly, for any $z_1, z_2$ such that $\socialV \leq z_1 \leq z_2$, $-z_2 \leq V_i < -z_1 \Rightarrow V_i + \socialV \leq 0$, thus
\begin{align*}
	\sw_i(z_1) - \sw_i(z_2) =&  \E{(V_i+\socialV) \cdot \one{V_i \geq -z_1}} - \E{(V_i+\socialV) \cdot \one{V_i \geq -z_2}} \\
 	=& - \E{(V_i+\socialV) \cdot \one{-z_2 \leq V_i < -z_1}} \geq 0.
\end{align*}
This completes the proof of this lemma.
\end{proof}

\subsection{Proof of Lemma~\ref{lem:crossing_utilities} } \label{appx:proof_lem_crossing_utilities}

\lemCrossingUtils*

\begin{proof} For part (i), we know from part (i) of Lemma~\ref{lem:util_welfare} that $u_i'(z+) = \ut_i(z) - 1$. Therefore:
\begin{align*}
	&u_1(z_2) - u_1(z_1) \leq u_2(z_2) - u_2(z_1)  \\
		\Rightarrow & \int_{z_1}^{z_2} (\ut_1(v)  - 1) dv \leq \int_{z_1}^{z_2} (\ut_2(v)  - 1) dv \\
		\Rightarrow & \int_{z_1}^{z_2} \ut_1(v) dv \leq \int_{z_1}^{z_2}  \ut_2(v) dv
\end{align*}
Since $\ut_1(z)$ and $\ut_2(z)$ are both non-negative and monotonically non-decreasing in $z$, we have:
\begin{equation*}
	(z_2 - z_1) \cdot \ut_1(z_1) \leq \int_{z_1}^{z_2} \ut_1(v) dv \leq \int_{z_1}^{z_2} \ut_2(v) dv \leq (z_2 - z_1) \cdot \ut_2(z_2).
\end{equation*}
Which implies $ \ut_2(z_2) \geq \ut_1(z_1)$.
For part (ii), first observe that 
\begin{align*}
	\sw_i(z) = \E{(V_i + \socialV)\one{V_i \geq -z}} = 
	u_i(z) + z(1 - \ut_i(z)) +  \socialV \ut_i(z).
\end{align*}
When $z_1 \leq z_2 \leq \socialV$, given $u_2(z_2) \geq u_1(z_2)$, $\ut_i(z) = u_i'(z+) + 1$ and the convexity of $u_1(z)$, 
\begin{align*}
	& \sw_2(z_2) - \sw_1(z_1)\\
	= & u_2(z_2) + z_2(1 - \ut_2(z_2)) +  \socialV \ut_2(z_2) - \left( u_1(z_1) + z_1(1 - \ut_1(z_1)) +  \socialV \ut_1(z_1)\right) \\
	\geq & u_1(z_2) - u_1(z_1) + z_2 - z_1  - z_2 \ut_2(z_2) + z_1 \ut_1(z_1) + \socialV(\ut_2(z_2) - \ut_1(z_1)) \\
	= & \int_{z_1}^{z_2} (ut_1(z) - 1) dz + z_2 - z_1  - z_2 \ut_2(z_2) + z_1 \ut_1(z_1) + \socialV(\ut_2(z_2) - \ut_1(z_1)) \\
	\geq & \int_{z_1}^{z_2} ut_1(z_1) dz - z_2 \ut_2(z_2) + z_1 \ut_1(z_1) + \socialV(\ut_2(z_2) - \ut_1(z_1)) \\
	= & ut_1(z_1)(z_2 - z_1) - z_2 \ut_2(z_2) + z_1 \ut_1(z_1) + \socialV(\ut_2(z_2) - \ut_1(z_1)) \\
	= & (\socialV - z_2) (\ut_2(z_2) - \ut_1(z_1)) \geq 0.
\end{align*}
When $\socialV \leq z_1 \leq z_2$, we have:
\begin{align*}
	& \sw_2(z_2) - \sw_1(z_1)  \\
	= & u_2(z_2) + z_2(1 - \ut_2(z_2)) +  \socialV \ut_2(z_2) - \left( u_1(z_1) + z_1(1 - \ut_1(z_1)) +  \socialV \ut_1(z_1)\right) \\
	\leq & u_2(z_2) - u_2(z_1) + z_2 - z_1  - z_2 \ut_2(z_2) + z_1 \ut_1(z_1) + \socialV(\ut_2(z_2) - \ut_1(z_1)) \\
	= & \int_{z_1}^{z_2} (ut_2(z) - 1) dz + z_2 - z_1  - z_2 \ut_2(z_2) + z_1 \ut_1(z_1) + \socialV(\ut_2(z_2) - \ut_1(z_1)) \\
	\leq & \int_{z_1}^{z_2} ut_2(z_2) dz - z_2 \ut_2(z_2) + z_1 \ut_1(z_1) + \socialV(\ut_2(z_2) - \ut_1(z_1)) \\
	= & ut_2(z_2)(z_2 - z_1) - z_2 \ut_2(z_2) + z_1 \ut_1(z_1) + \socialV(\ut_2(z_2) - \ut_1(z_1)) \\
	= & (\socialV - z_1) (\ut_2(z_2) - \ut_1(z_1)) \leq 0.
\end{align*}
This completes the proof of this lemma. 
\end{proof}

\subsection{The First Best Social Welfare} \label{appx:proof_fb_sw}

\begin{proposition}[First Best Welfare] \label{prop:fb_sw} For any agent $i$ with type $\CDF_i$ satisfying (A1)-(A2), the first best social welfare subject to IR and ND constraints is either achieved by charging a penalty $z = \socialV$, or by charging the highest penalty within the intersection of her IR and ND ranges.
\end{proposition}

\begin{proof}
From the monotonicity properties proved in Lemma~\ref{lem:util_welfare}, we know that if there exists $y \in \setR$ s.t. $u_i(z) - y \geq 0$ and $\rev_i(z,y) \geq 0$, meaning that the intersection of the agent's IR and ND ranges include $(\socialV, y)$ for some $y \in \setR$, then the first best social welfare is achieved by charging the agent such a two-part payment where the penalty is $z = \socialV$. What is left to consider is the case where the intersection of the IR and ND range does not include penalty $z = \socialV$. We show that if the intersection of the IR and ND ranges includes some $(z,y)$, then for any $z' \in [0, z]$, there exists $y' \in \setR$ s.t. $(z', y')$ is also in the intersection of the IR and ND ranges. This implies that when $z = \socialV$ is not included, then the highest achievable welfare must be achieved by the highest penalty in the intersection. Note that given $u_i(z) - y \geq 0$ and $\rev_i(z,y) \geq 0$, we have
\begin{align*}
	u_i(z) - y + \rev_i(z,y) = \E{V_i \one{V_i \geq -z}}  \geq 0.
\end{align*}
From the monotonicity of $\E{V_i \one{V_i \geq -z}}$, we know that for any $z' \in [0, z]$, $u_i(z') - y' + \rev_i(z',y') = \E{V_i \one{V_i \geq -z'}}  \geq 0$ holds for any $y' \in \setR$, and as a result, it is possible to find $y'$ s.t. $u_i(z') - y' \geq 0$ and $\rev_i(z',y') \geq 0$ both hold. 
\end{proof}

\subsection{Proof of Lemma~\ref{lem:lem_P1P5_characterization}} \label{appx:proof_P1P5}

Before proving the lemma, we first provide a few useful results. 
The first result proves some more properties of the expected utility function $u_i(z)$: for any type that satisfies (A1) and (A2), $u_i(z)$ always resides above $-z$, and converges to $-z$ as $z \rightarrow -\infty$. 

\begin{lemma}
\label{lem:exp_u_appx} Assuming (A1) and (A2), the expected utility $u_i(z)$ as a function of penalty $z$ satisfies:

\begin{enumerate}[(i)]
	\setlength\itemsep{0em}
	\item $u_i(z) + z$ is non-negative, and is monotonically non-decreasing in $z$.
	\item $\lim_{z \rightarrow -\infty} u_i(z) + z = 0	$.
\end{enumerate}
\end{lemma}

\begin{proof} 

Define $g_i(z) \triangleq u_i(z) + z$, we know that $g_i(z)$ is continuous and convex. For part (i), the non-negativity $g_i(z) \geq 0$ holds, since $u_i(z) = \E{\max \{V_i,-z\}} \geq -z$. The monotonicity holds, since Lemma~\ref{lem:util_welfare} implies that the right derivative of $g_i(z)$ is given by $g_i'(z+) = u_i'(z+) + 1 = ut_i(z) = \Pm{V_i \geq -z}$, which is non-negative for all $z \in \setR$. Part (ii) can now be rewritten as $\lim_{z \rightarrow -\infty} g_i(z)  = 0$. Given the non-negativity and the monotonicity of $g_i(z)$, we know that if part (ii) does not hold, then there exists some $\eps > 0$, s.t. $g_i(z) \geq \eps$ for all $z \in \setR$. We show that this results in a contradiction.

First, note that for any $z < 0$, $g_i(z) = \E{ \max \{V_i, -z\} } + z = \E{ \max \{V_i^+, -z \} } + z$.  Given assumption (A2), we know that $\E{V_i^+}$ is finite. This implies that $\lim_{z \rightarrow -\infty} \E{\min\{ V_i^+, -z\}} = \E{V_i^+}$ (from monotone convergence theorem), and that for the positive constant $\eps/2 > 0$, there exists a large constant $\largeCnst \in \setR_{\geq 0}$ s.t. $\forall z < -\largeCnst$, $\E{\min\{ V_i^+, -z\}} \geq \E{V_i^+} - \eps/2$. This is a contradiction, since for any $z < -\largeCnst$, we have $\E{V_i^+} = \E{\max \{ V_i^+, -z\}} + \E{\min \{ V_i^+, -z\}} - (-z)  = u_i(z) + z + \E{\min \{ V_i^+, -z\}} \geq \eps + \E{V_i^+} - \eps/2 = \E{V_i^+} + \eps/2 \Rightarrow \eps \leq 0$. 

This completes the proof of this lemma. 
\end{proof}

The following lemma shows that there exist valid agent types satisfying (A1)-(A2) that correspond to expected utility functions with certain properties.

\begin{lemma} \label{lem:u_to_F} For any function $u:\setR \rightarrow \setR$ that is continuous, convex, monotonically decreasing, and satisfies $\lim_{z \rightarrow -\infty} u(z) + z = 0$ and $u(0) > 0$, there exists a random variable $V$ with CDF $\CDF$ such that $\CDF$ satisfies (A1) and (A2), and that $u(z) = \E{\max\{V, -z\}}$ for all $z \in \setR$. Moreover, $\CDF$ also satisfies (A3) if $\lim_{z \rightarrow +\infty}u(z) < 0$. 

\end{lemma}

\begin{proof} We prove this lemma by constructing an agent type $\CDF$ from $u(z)$, and showing that  it has the desired properties. Given that $u$ is convex, we know that it is semi-differentiable. Let  $\CDF$ be defined as $\CDF(v) = - u'((-v)-)$, i.e. the negation of the left derivative of $u$ at $-v$. We show that $\CDF$ is a valid CDF that satisfies (A1), (A2), and for a random variable $V$ with distribution $\CDF$, $u(z) = \E{\max\{V, -z\}}$ holds for all $z \in \setR$.

We first show that $\CDF$ is a valid cumulative distribution function, i.e. $\CDF(v)$ is monotonically increasing, right-continuous, lower bounded by $0$ and upper bounded by $1$. The monotonicity is obvious, since $u$ is convex, thus $u'(v-)$ is non-decreasing in $v$, thus $-u'((-v)'-)$ is also non-decreasing in $v$. The right continuity of $\CDF$ is implied by the left-continuity of $u'(v-)$, and $u'(v-)$ is left-continuous since the left derivative of convex function is monotone and cannot have jump discontinuities. $u$ being non-increasing implies that $u'(v) \leq 0$ for all $v \in \setR$ thus $F(v) \geq 0$ for all $v \in \setR$. Finally, if $F(v) > 1$ for some $v \in \setR$, we know that $u'((-v)-) < -1$ holds for some $v$, and as a result of the convexity of $u$, $u(z) + z$ must diverge as $z \rightarrow -\infty$, which contradicts the assumption that $\lim_{z \rightarrow -\infty} u(z) + z = 0$. 

We now show that the type $F$ satisfies (A1) and (A2), i.e. $\E{V^+} > 0$ and $\E{V^+} < +\infty$. Since $V^+$ is a non-negative random variable where $\Pm{V^+ \leq v} = F(v)$ for all $v \geq 0$, we have
\begin{align*}
	\E{V^+} =& \int_0^{+\infty} (1- \Pm{V^+ \leq v}) dv =  \int_0^{+\infty} (1- F(v)) dv = \int_0^{+\infty} (1 + u'((-v)-)) dv 
\end{align*}
Define $g(v) \triangleq u(-v)$ for all $v \in \setR$, we know that $g'(v+) = -u'((-v)-)$, therefore
\begin{align*}
	&\int_0^{+\infty} (1 + u'((-v)-)) dv  = \lim_{z \rightarrow + \infty} \int_{0}^z (1 + u'((-v)-))dv  
	= \lim_{z \rightarrow + \infty} \int_{0}^z (1 -g'(v+) )dv \\  
	= & \lim_{z \rightarrow +\infty}  z - g(z) + g(0)  
	=  \lim_{z \rightarrow +\infty} z - u(-z) + u(0) = u(0) - \lim_{z \rightarrow -\infty} (u(z)+z) = u(0),
\end{align*}
which implies that $\E{V^+} = u(0) > 0$ and is finite. What is left to prove is $u(z) = \E{\max\{V, -z\}}$. This is obvious, since $u(z)$ as a function has the same right derivative as the right derivative of $\E{\max\{V, -z\}}$ w.r.t. $z$, and we had just shown above that the two functions coincide at $z = 0$.

Now assume that $\lim_{z \rightarrow +\infty}u(z) < 0$. We know that there exists $z \in \setR$ s.t. $u(z) < 0$. As a result, $\E{V}  \leq \E{\max\{V, -z\}} = u(z) < 0 \Rightarrow \E{V} < 0$, thus (A3) is satisfied. 
This completes the proof of this lemma.
\end{proof}

We now formally define first order stochastic dominance and positive responsiveness.  

\begin{definition}[First order stochastic dominance] \label{defn:fosd} Let $\CDF_1$ and $\CDF_2$ be two value distributions. $\CDF_1$ first-order stochastic dominates $\CDF_2$, which we denote $\CDF_1 \succeq_{FOSD} \CDF_2$, if $\CDF_1(v) \leq \CDF_2(v)$ for all $v \in \setR$ and $\CDF_1(v) < \CDF_2(v)$ for some $v \in \setR$. The dominance is strict if $\CDF_1(v) < \CDF_2(v)$ for all $ v \in \setR$ s.t. $\CDF_1(v) < 1$ and $\CDF_2(v) > 0$, in which case we denote $\CDF_1 \succ_{FOSD} \CDF_2$.
\end{definition}

\begin{definition}[Positive Responsiveness] \label{defn:PR} A two-period mechanism with DSE $\report^\ast$ is \emph{positively responsive} (PR) if for any economy $\CDF$, 
and each agent $i$, agent $i$ getting allocated with probability $\alloc \in (0,1)$ under the DSE implies that if her type was $\CDF'_i$ s.t. $\CDF'_i \succ_{FOSD} \CDF_i$, then reporting $\report^\ast(\CDF'_{i})$,  she gets allocated with probability one.
\end{definition}

Intuitively, the PR condition requires that if an agent is allocated with some positive probability that is less than one, then if she instead has a strictly ``higher" type in terms of FOSD, then she will get allocated with probability under the DSE of the mechanism.

The following lemma proves that when an agent is allocated with probability less than one, then she cannot get strictly positive utilities. 

\begin{lemma} \label{lem:zero_utility_when_tied} Assume that the type space includes all value distributions satisfying (A1)-(A3). Under any two-period mechanism that satisfies (P1)-(P6) and under the DSE, agents who are allocated with probability less than one get zero expected utilities. 
This result still holds if the type space is the set of all $(\fixedV_i, \fixedP_i)$ types. 
\end{lemma}

\begin{proof} 
Given assumption (P2) IR and (P6) no payment to unassigned agents, we know that agents who gets assigned with probability zero gets zero expected utilities. 
Now assume toward a contradiction, that there exists a mechanism that satisfies (P1)-(P6) and an economy $\CDF$, where there exists an agent $i \in N$ who gets positive expected utility, but is only allocated with probability $\alloc_i \in (0,1)$. In the events that she is allocated, denote $z_i \triangleq \tone_i(\report^\ast(\CDF))$ as her penalty and $y_i \triangleq \tzero_i(\report^\ast(\CDF))$ as her base payment. Agent $i$ getting positive utility requires that $\alloc_i(u_i(z_i) - y_i)  = \pi$ for some $\pi > 0$, i.e. payment $(z_i,y_i)$ resides on the curve $u_i(z_i) - \pi/\alloc_i$. See Figure~\ref{fig:dominated_general_v}. %
We first show a contradiction for general types that satisfy (A1)-(A3), then for the $(\fixedV_i, \fixedP_i)$ types.

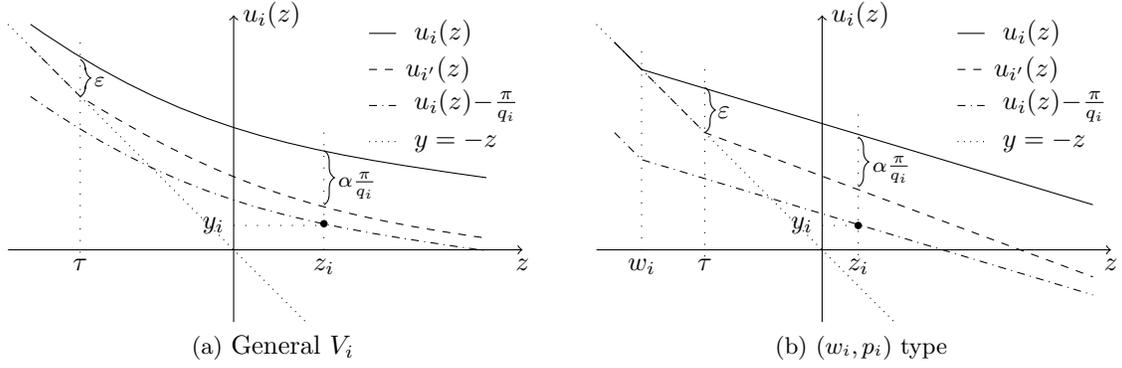
\begin{figure}[t!]
\centering 
\subfloat[\small{General $V_i$}]{\label{fig:dominated_general_v}
\begin{tikzpicture}[scale = 1.2][font=\small]
\draw[->] 	(-2.5,0) -- (3.2, 0) node[anchor=north] {$z$};
\draw[->] 	(0,-0.8) -- (0, 2.6) node[anchor=west] {$u_i(z)$};

\draw[dotted] (-2.5, 2.5) -- (0.8, -0.8);

\draw[-, name path = uA] (-2.25, 2.5) to [out=-35, in = 172] (2.8, 0.8);
\draw[dashdotted, name path = uA] (-2.25, 1.7) to [out=-35, in = 172] (2.8, -0.01);
\draw[dashed] (-2.25, 2.25)--	(-1.7, 1.7)  to [out=-32, in = 172] (2.8, 0.13);

\draw [decorate,decoration={brace,amplitude=4pt},xshift=0cm,yshift=0pt]
      (-1.7,2.11) -- (-1.7,1.7) node [midway,right,xshift=-0.05cm,yshift=-0.05cm] {{ \scriptsize 
      $\eps$}};

\draw [decorate,decoration={brace,amplitude=4pt},xshift=0cm,yshift=0pt]
      (1,1.1) -- (1,0.5) node [midway,right,xshift=-0.05cm,yshift=-0.1cm] {{ \scriptsize  $\alpha \frac{\pi}{\alloc_i}$}};

\filldraw [black] (1, 0.29) circle (1pt);
\draw[loosely dotted](1, 1.3) -- (1,-0.1);
\draw (1, 0) node[anchor=north] {$z_i$};
\draw[loosely dotted] (0, 0.27) -- (1, 0.27);
\draw (0, 0.27) node[anchor=east] {$y_i$};

\draw[-] (1.5, 2.4) -- (1.8, 2.4) node[anchor=west] {{ $u_i(z)$}};
\draw[dashed] (1.5, 2) -- (1.8, 2) node[anchor=west] {{$u_{i'}(z)$}};
\draw [dashdotted](1.5, 1.6) -- (1.8, 1.6) node[anchor=west] {{  $u_i(z) \hspace{-0.2em} - \hspace{-0.2em} \frac{\pi}{\alloc_i}$}};
\draw [ dotted](1.5, 1.2) -- (1.8, 1.2) node[anchor=west] {{ $y = -z$}};

\draw[loosely dotted](-1.7, 2.3) -- (-1.7,-0.1);
\draw (-1.7, 0) node[anchor=north] {$\tau$};

\end{tikzpicture}
}
\hspace{1em}
\subfloat[$(\fixedV_i,\fixedP_i)$ type]{\label{fig:dominated_wipi}
\begin{tikzpicture}[scale = 1.2][font=\small]
\draw[->] 	(-2.5,0) -- (3.2, 0) node[anchor=north] {$z$};
\draw[->] 	(0,-0.8) -- (0, 2.6) node[anchor=west] {$u_i(z)$};

\draw[dotted] (-2.5, 2.5) -- (0.8, -0.8);

\draw[-] 	(-2.3, 2.3) -- (-2.0, 2.0) -- (3, 0.5);
\draw[dashed] (-2.3, 2.3)--	(-1.3, 1.3) -- (3, 0.5 - 0.8); 
\draw[dashdotted] (-2.3, 2.3-1) -- (-2.0, 2.0-1) -- (3, 0.5-1);

\draw [decorate,decoration={brace,amplitude=4pt},xshift=0cm,yshift=0pt]
      (-1.3,1.8) -- (-1.3,1.3) node [midway,right,xshift=.05cm,yshift=.-0.03cm] {$\eps$};

\draw [decorate,decoration={brace,amplitude=4pt},xshift=0cm,yshift=0pt]
      (0.4,1.25) -- (0.4,0.7) node [midway,right,xshift=-0.05cm,yshift=-0.1cm] {{ \scriptsize  $\alpha \frac{\pi}{\alloc_i}$}};

\draw[-] (1.5, 2.4) -- (1.8, 2.4) node[anchor=west] {{ $u_i(z)$}};
\draw[dashed] (1.5, 2) -- (1.8, 2) node[anchor=west] {{$u_{i'}(z)$}};
\draw [dashdotted](1.5, 1.6) -- (1.8, 1.6) node[anchor=west] {{  $u_i(z) \hspace{-0.2em} - \hspace{-0.2em} \frac{\pi}{\alloc_i}$}};
\draw [ dotted](1.5, 1.2) -- (1.8, 1.2) node[anchor=west] {{ $y = -z$}};

\draw[loosely dotted](-2, 2.2) -- (-2,-0.1);
\draw (-2, 0) node[anchor=north] {$\fixedV_i$};

\draw[loosely dotted](-1.3, 2) -- (-1.3,-0.1);
\draw (-1.3, 0) node[anchor=north] {$\tau$};

\filldraw [black] (0.4, 0.27) circle (1pt);
\draw[loosely dotted](0.4, 1.5) -- (0.4,-0.1);
\draw (0.4, 0) node[anchor=north] {$z_i$};
\draw[loosely dotted] (0, 0.27) -- (0.5, 0.27);
\draw (0, 0.27) node[anchor=east] {$y_i$};

\end{tikzpicture}
}
\caption{Illustrations for the proof of lemma~\ref{lem:zero_utility_when_tied}.
\label{fig:dominated_rv}} 
\end{figure}

\bigskip

\noindent{}\emph{Step 1. General types under (A1)-(A3).} We first construct the expected utility function $u_{i'}(z)$ as shown in the dashed line in Figure~\ref{fig:dominated_general_v}. First, let $\eps$ be defined as $\eps \triangleq \min \left\lbrace \frac{u_i(0)}{2}, ~\frac{\pi}{2\alloc_i} \right\rbrace$, we know $\eps > 0$. Moreover, let $\beta$ be defined as $\beta \triangleq \eps \alloc_i/\pi$ (therefore $\eps = \beta \pi/\alloc_i$ holds), we know that there are two cases: when $u_i(0) \geq \pi/\alloc_i$ (as shown in Figure~\ref{fig:dominated_general_v}), we have $\beta = 1/2$, and when $u_i(0) < \pi/\alloc_i$, in which case $\beta < 1/2$. Now, define $\tau$ as the point where $u_i(z) - \eps$ and the $y = -z$ cross each other, i.e.  
\begin{align}
	\tau \triangleq \inf \{ z \in \setR~|~ u_i(z) \geq -z + \eps \}, \label{equ:infimum_piqi}
\end{align}
we show that $\tau$ exists and is finite. 
We know from part (ii) of Lemma~\ref{lem:exp_u_appx} that the set $\{ z \in \setR~|~ u_i(z) \geq -z + \eps \}$ is bounded from below, since $u_i(z) + z$ approaches zero as $z \rightarrow -\infty$. We also know that $\{ z \in \setR~|~ u_i(z) \geq -z + \eps \}$ is not empty, since $\eps \leq u_i(0)/2$, therefore $u_i(0) > -0 + \eps$, which implies $0 \in \{ z \in \setR~|~ u_i(z) \geq -z + \eps \}$. Therefore, the infimum $\tau$ is finite, and we know from the monotonicity and continuity of $u_i(z) + z$ that $u_i(\tau) + \tau = \eps$, and that $\tau < 0$.  Moreover, $u_i'(\tau+) > -1$ must hold, since otherwise we must have $u_i(\tau) = -\tau$ in order not to violate Lemma~\ref{lem:exp_u_appx}. 
Now we need to consider two cases, depending on whether $\tau < z_i$ holds. 

\medskip

\noindent{}\emph{Case 1: $\tau < z_i$.}
We first consider the case where $\tau < z_i$, as shown in Figure~\ref{fig:dominated_general_v}. Define $\alpha$ as
\begin{align*}
	\alpha \triangleq \min \left\lbrace \frac{3}{4},~ \left(\frac{1}{2} + \frac{1}{2\alloc_i} \right) \beta \right\rbrace,
\end{align*}
we know $\alpha > \beta$ and $\alpha < \beta/\alloc_i$, since $1/\alloc_i > 1$. 
Moreover, define $\delta$ as
\begin{align*}
	\delta = \min\{\delta_1,~\delta_2,~\delta_3 \}, \text{ where } \delta_1 = -\frac{u_i(0)}{4\tau},~ \delta_2 = \frac{ (\alpha - \beta) \pi/\alloc_i}{z_i - \tau},\txtand \delta_3 = u_i'(\tau+) + 1,
\end{align*}
we know $\delta$ is strictly positive. 
Now, define the utility function $u_{i'}(z)$ as the following:
\begin{align}
	u_{i'}(z) = \pwfun{-z, & \txtif z < \tau \\
		u_i(z) - \eps - \delta(z - \tau), & \txtif z \geq \tau}, \label{equ:constructed_uprime}
\end{align}
we know that
\begin{align*}
	u_{i'}(0) = u_i(0) - \eps + \delta \tau  \geq u_i(0) - u_i(0)/2  - u_i(0)/4 = u_i(0)/4 > 0. 
\end{align*}
Moreover, it is easy to check that $u_{i'}(z)$ is continuous, convex, monotonically decreasing, and satisfies $\lim_{z \rightarrow -\infty} u_{i'}(z) + z = 0$. Lemma~\ref{lem:u_to_F} then implies that $u_{i'}(z)$ is the expected utility of some agent whose type satisfies (A1) and (A2). Let's denote the CDF as $\CDF_{i'}$ and the random value of this agent as $V_{i'}$. Assuming that (A3) is satisfied by $\CDF_i$, we know $\lim_{z \rightarrow \infty} u_{i'}(z) <  \lim_{z \rightarrow \infty} u_{i}(z) < 0$, thus $\CDF_{i'}$ also satisfies (A3). It is also easy to show that $\CDF_i \succ_{FOSC} \CDF_{i'}$, since for all $v > -\tau$, $\CDF_{i'}(v) = -u'_{i'}((-v)-) = 1$, which is strictly greater than $\CDF_{i}(v)$ if $\CDF_i(v) < 1$, and for all $v \leq -\tau$, we have $\CDF_{i'}(v) = -u'_{i'}((-v)-) = - (u'_{i}((-v)-) - \delta) = \CDF_{i}(v) + \delta > \CDF_{i}(v)$. 

Given the construction of $u_{i'}(z)$, we can also check that 
\begin{align*}
	u_{i'}(z_i) - y_i = u_i(z_i)- y_i - \eps - \delta(z_i - \tau) \geq 
	 \frac{\pi}{\alloc_i} - \beta \frac{\pi}{\alloc_i} - (\alpha - \beta) \frac{\pi}{\alloc_i} = (1 - \alpha) \frac{\pi}{\alloc_i}  > 0.
	%
\end{align*}
As a result, in the economy where agent $i$ is replaced by agent $i'$ with type $\CDF_{i'}$, if agent $i'$ pretends that her type is actually $\CDF_i$ and reports $\report_i^\ast(\CDF_{i})$, we know that she is allocated with probability $\alloc_i$, charged $z_i$ as penalty and $y_i$ as her base payment if allocated, thus her expected utility would be
\begin{align*}
	\alloc_i (u_{i'}(z_i) - y_i ) = (1-\alpha) \pi > 0.
\end{align*}
Therefore, in the economy $(\CDF_{i'}, \CDF_{-i})$, agent $i'$ must get expected utility at least $(1-\alpha) \pi$ --- otherwise, pretending to have type $\CDF_{i}$ would be a useful deviation. We also know that agent $i'$ cannot be tied with any other agent, since if she is tied, i.e. getting allocated with some positive probability that's smaller than one, then positive responsiveness requires that agent $i$ to be allocated with probability $1$ in the economy $(\CDF_i, \CDF_{-i})$, which contradicts the assumption that agent $i$ is tied with some other agent. Let $z_{i'} \triangleq \tone_{i'}(\CDF_{i'}, \CDF_{-i})$ and $y_{i'} \triangleq \tzero_{i'}(\CDF_{i'}, \CDF_{-i})$, we know that $u_{i'}(z_{i'}) \geq y_{i'} + (1-\alpha) \pi$ holds.

Now consider two cases: (i) $z_{i'} < \tau$, and (ii) $z_{i'} \geq \tau$. When $z_{i'} < \tau$, we know that agent $i'$ never uses the resource (since $\Pm{V_{i'} \geq -z_{i'}} = 0$). As a result, the agent getting expected utility at least $(1-\alpha) \pi$ means that in expectation the mechanism pays the agent $(1-\alpha) \pi > 0$, which violates ND. For case (ii), we know by construction that $u_{i}(z_{i'}) \geq u_{i'}(z_{i'}) + \eps$ holds, therefore, in the original economy, if agent $i$ pretends that her type is in fact $\CDF_{i'}$, she gets expected utility:
\begin{align*}
	u_{i}(z_{i'}) - y_{i'} \geq u_{i'}(z_{i'}) + \eps - y_{i'} \geq (1-\alpha) \pi + \beta \frac{\pi}{\alloc_i}  > \pi.
\end{align*}
The last inequality holds since $\alpha < \beta/\alloc_i $ by construction. This means that pretending to be of type $\CDF_{i'}$ is a useful deviation for agent $i$, therefore violating (P1) DSE. Since neither of (i) or (ii) can hold, this is a contradiction, and completes the proof for Case 1 $\tau < z_i$.

\medskip

\noindent{}\emph{Case 2: $\tau \geq z_i$.}
For the case where $\tau \geq z_i$, define $u_{i'}(z)$ in the same way as in \eqref{equ:constructed_uprime}, but for $\delta \triangleq (u_i'(\tau+) + 1)/2$. We can similarly show that $u_{i'}(z)$ as constructed is the utility function corresponding to a valid agent type $\CDF_{i'}$ satisfying (A1)-(A3) which is strictly first-order stochastic dominated by $\CDF_i$. 
Note that $u_{i'}(z) \geq u_i(z) - \eps$ holds for all $z \leq \tau$. 
In the original economy, replacing $i$ with $i'$, we know that if agent $i'$ pretends to have type $\CDF_i$, her expected utility is at least $\alloc_i (u_i'(z_i) - y_i) \geq \alloc_i(u_i(z_i)- \eps - y_i) \geq \pi - \eps \alloc_i \geq  \pi/2 > 0$.  Therefore, in the economy $(\CDF_{i'}, \CDF_{-i})$, agent $i'$ must be allocated the resource with probability one and get utility at least $ \pi - \eps \alloc_i$, in order not to violate DSE and positive-responsiveness. Moreover, similar to the above case, the payment she is charged $(z_{i'}, y_{i'})$ has to satisfy $z_{i'} \geq \tau$ in order not to violate ND. Now, if agent $i$ gets the outcome of agent $i'$ instead, her utility will be
\begin{align*}
	u_i(z_{i'}) - y_{i'} \geq u_{i'}(z_{i'})+ \eps - y_{i'} \geq \pi - \alloc_i \eps + \eps > \pi,
\end{align*}
and this is a useful deviation for agent $i$. This completes the proof of Step 1.

\bigskip

\noindent{}\emph{Step 2. $(\fixedV_i,\fixedP_i)$ types.} 
To complete the proof of this theorem, what is left to show is that when agent $i$'s type follows the $(\fixedV_i, \fixedP_i)$ model, then the constructed type $\CDF_{i'}$ in the two cases in Step 1 also follows the $(\fixedV_i, \fixedP_i)$ model. This obvious, since assuming that $\Pm{V_i = \fixedV_i} = \fixedP_i$ and $\Pm{V_i = -\infty} = 1 - \fixedP_i$, the random variable $V_{i'}$ corresponding to the expected utility function $u_{i'}(z)$ (see Figure~\ref{fig:dominated_wipi}) satisfies $V_{i'} = -\tau$ with probability $\fixedP_i - \delta > 0$, and $V_{i'} = -\infty$ with probability $1 - \fixedP_i + \delta$. 
\end{proof}

\medskip

The last lemma characterizes the range of admissible two-part payments for any two-period mechanism that satisfies (P1)-(P6).

\begin{lemma}[Admissible payments] \label{lem:admissible_payments}
Assume the type space includes all $(\fixedV_i, \fixedP_i)$ types, and fix any two-period mechanism that satisfies (P1)-(P6). For any type profile $\CDF$, the two-part payment $(z^\ast, y^\ast)$ that the allocated agent $\winner$ is charged satisfy: (i) $y^\ast \geq 0$, and (ii) $y^\ast \geq -z^\ast$. 
\end{lemma}

\begin{proof} Consider an agent $\winner$ who is allocated with probability $\alloc_\winner > 0$, and is charged  $(z^\ast, y^\ast)$ when allocated. Agent $\winner$'s expected utility is $\alloc_\winner (u_\winner(z^\ast) - y^\ast)$. 

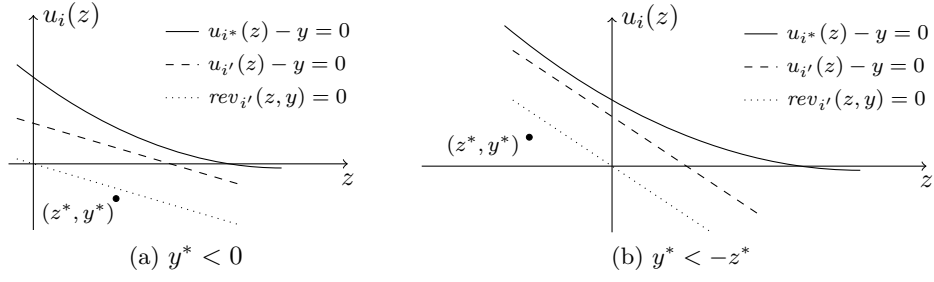
\begin{figure}[t!]
\centering 
\subfloat[\small{$y^\ast < 0$}]{\label{fig:negative_y}
\begin{tikzpicture}[scale = 1.1][font=\small]
\draw[->] 	(-0.3,0) -- (3.8,0) node[anchor=north] {$z$};
\draw[->] 	(0,-0.8) -- (0,1.8) node[anchor=west] {$u_i(z)$};

\draw[-] 	(-0.2, 1.2) parabola[bend at end] (3, -0.05); 
\draw[dashed] (-0.2, 0.55) -- (2.5, -0.25);
\draw[dotted] (-0.2, 0.06) -- (2.5, -0.25-0.49);

\filldraw [black] (1, -0.42) circle (1pt);
\draw (1.1, -0.6) node[anchor = east]{{\scriptsize $(z^\ast,y^\ast)$}};

\draw[-] (1.6, 1.6) -- (2, 1.6) node[anchor=west] {{\scriptsize $u_{i^\ast}(z) -y = 0$}};
\draw[dashed] (1.6, 1.2) -- (2, 1.2) node[anchor=west] { {\scriptsize $u_{i'}(z) -y = 0$}};
\draw [dotted](1.6, 0.8) -- (2, 0.8) node[anchor=west] {{\scriptsize $\rev_{i'}(z,y) = 0$}};

\end{tikzpicture}
}
\hspace{1em}
\subfloat[$y^\ast < -z^\ast$]{\label{fig:y_less_than_minus_z}

\begin{tikzpicture}[scale = 1.1][font=\small]
\draw[->] 	(-2.3,0) -- (3.8, 0) node[anchor=north] {$z$};
\draw[->] 	(0,-0.8) -- (0, 1.8) node[anchor=west] {$u_i(z)$};

\draw[-] 	(-1.3, 1.7) parabola[bend at end] (3, -0.05); 
\draw[dashed] (-1.2, 1.4) -- (1.8, -0.6);
\draw[dotted] (-1.2, 0.8) -- (1.2, -0.8);

\filldraw [black] (-1, 0.35) circle (1pt);
\draw (-1, 0.25) node[anchor = east]{{\scriptsize $(z^\ast,y^\ast)$}};

\draw[-] (1.6, 1.6) -- (2, 1.6) node[anchor=west] {{\scriptsize $u_{i^\ast}(z) -y = 0$}};
\draw[dashed] (1.6, 1.2) -- (2, 1.2) node[anchor=west] { {\scriptsize $u_{i'}(z) -y = 0$}};
\draw [dotted](1.6, 0.8) -- (2, 0.8) node[anchor=west] {{\scriptsize $\rev_{i'}(z,y) = 0$}};

\end{tikzpicture}
}
\caption{Illustration for the proof of Lemma~\ref{lem:admissible_payments} admissible payments.
\label{fig:proof_lemma_admissible_payments}} 
\end{figure}

\medskip

\noindent{}\textit{Part (i): $y^\ast \geq 0$.} Assume towards a contradiction, that there exists an economy, where the allocated agent $\winner$ is charged a two part payment $(z^\ast, y^\ast)$ where the base payment $y^\ast < 0$. First, $z^\ast > 0$ must hold: for the case where $\alloc_\winner < 1$, i.e. when agent $\winner$ is tied with some other agent, Lemma~\ref{lem:zero_utility_when_tied} implies that $\alloc_\winner (u_i(z^\ast) - y^\ast) = 0$, and this can only be the case if $z^\ast > 0$; for the case where $\alloc_\winner = 1$, if $z^\ast \leq 0$, the expected revenue of the mechanism is negative and this violates ND. Therefore, we have a situation as shown in Figure~\ref{fig:negative_y}, where the solid line depicts the zero-profit curve of agent $\winner$, and the dot indicates the two-part payment in the two-dimensional payment space.

Now consider a new agent $i'$ with random value $V_{i'}$ that follows the $(\fixedV_i, \fixedP_i)$ type model, where the parameters are given by 
\begin{align*}
	\fixedP_{i'} = & \max \left\lbrace  1 + \left( \frac{1}{4} \alloc_\winner + \frac{1}{2} \right) \frac{y^\ast}{z^\ast} ,~ 1/2 \right\rbrace, \\
	\fixedV_{i'} = & -  \frac{1}{\fixedP_{i'}} \frac{\alloc_\winner }{4}  y^\ast 
\end{align*}
Since $y^\ast < 0$ and $z^\ast > 0$, we know that $\fixedP_{i'}  \in [1/2,1)$, and that $\fixedV_{i'}> 0$, thus this is a valid agent type.  For any $z \geq 0$, the expected utility function of agent $i'$ is given by $u_{i'}(z) = \fixedV_{i'} \fixedP_{i'} - z(1 - \fixedP_{i'})$, thus the zero-profit curve of agent $i'$ is the dashed line in Figure~\ref{fig:negative_y}. Note that given this construction, $\fixedV_{i'}\fixedP_{i'} = - \alloc_\winner y^\ast /4$. Also note that $(z^\ast, y^\ast)$ is below the agent's budget balance curve $\rev_{i'}(z,y) = 0$ (the dotted line): $\rev_{i'}(z^\ast, y^\ast) = y^\ast + z^\ast (1 - \fixedP_{i'}) \leq y^\ast + z^\ast ( 1 - (1 + (\frac{1}{4} \alloc_\winner + \frac{1}{2}) \frac{y^\ast}{z^\ast}  ) ) = (1/2 - \alloc_\winner /4  )y^\ast < 0$.


Now consider the economy, where the reports of the rest of the agents are fixed, however, the type of agent $\winner$ is replaced with agent $i'$. We know that if agent $i'$ got the outcome of agent $\winner$, her expected utility is going to be:
\begin{align*}
	\alloc_\winner (u_{i'}(z^\ast)- y^\ast) = & \alloc_\winner ( \fixedV_{i'} \fixedP_{i'} - (1-  \fixedP_{i'}) z^\ast - y^\ast) \\
	 \geq &  \alloc_\winner \left( -\alloc_\winner \frac{y^\ast}{4} + \left( \frac{1}{4} \alloc_\winner + \frac{1}{2} \right) \frac{y^\ast}{z^\ast}   z^\ast - y^\ast \right) \\
	 =&  - \alloc_\winner y^\ast/2,
\end{align*}
%
%
which is strictly positive. Therefore, if agent $i'$ is not allocated, or if she is tied with some other agent (in which case she gets zero utility given Lemma~\ref{lem:zero_utility_when_tied}), or if she is allocated and gets utility lower than $ - \alloc_\winner y^\ast/2 $, she will have an incentive to report agent $\winner$'s type and get a higher expected utility. This violates IC, therefore, agent $i'$ must be assigned the resource with probability one, and let $(z_{i'}, y_{i'})$ be the two-part payment that she is charged. For $(\fixedV_i, \fixedP_i)$ types, we know that for all $z \geq 0$, the social surplus is fixed and is equal to the sum of agent's expected utility and mechanism's expected revenue:
\begin{align*}
	\sw_{i}(z) = \fixedV_{i'} \fixedP_{i'} = u_{i'} (z_{i'})- y_{i'} + \rev_{i'}(z_{i'}, y_{i'}) 
\end{align*}
As a result, $u_{i'} (z_{i'})- y_{i'} \geq - \alloc_\winner y^\ast/2 $ implies that $ \rev_{i'}(z_{i'}, y_{i'})  \leq \fixedV_{i'} \fixedP_{i'} + \alloc_\winner y^\ast/2  = - \alloc_\winner y^\ast/4 + \alloc_\winner y^\ast/2  < 0$. This violates ND, thus we conclude $y^\ast \geq 0$ must hold. 

\medskip

\noindent{}\textit{Part (ii): $y^\ast \geq -z^\ast$.} Given part (i), we know that if $z^\ast \geq 0$, $y^\ast \geq 0 \geq -z^\ast$ must  hold. Therefore, we only need to show that when an allocated agent $\winner$ is charged $(y^\ast, z^\ast)$ where $z^\ast < 0$ and $y^\ast < -z^\ast$,  then there is a contradiction. Similar to the previous case, we need to analyze whether agent $\winner$ is tied with any other agent. If $\winner$ is tied, i.e. allocated with probability $\alloc_\winner < 1$, she cannot get positive utility according to Lemma~\ref{lem:zero_utility_when_tied}. However, with $z^\ast < 0$ and $y^\ast < - z^\ast$, the agent's expected utility is positive (since Lemma~\ref{lem:exp_u_appx} implies that $u_\winner(z^\ast) - y^\ast \geq -z^\ast - y^\ast > 0$).

Therefore, the only case to consider is that agent $\winner$ is allocated the resource with probability one. In this case, $y^\ast \leq 0$ cannot hold without violating ND, therefore, the two-part payment $(z^\ast, y^\ast)$ charged must be as shown in Figure~\ref{fig:y_less_than_minus_z}.
Consider now an agent with value $V_{i'}$ that follows the $(\fixedV_i,\fixedP_i)$ type, where $\fixedV_{i'} > (-z^\ast)$ and $\fixedP_{i'} \in (0, 1 + y^\ast/z^\ast)$. Given $y^\ast < -z^\ast$, we know that $1 + y^\ast/z^\ast \in (0,1)$. With the same analysis as above, we know that when $z > -\fixedV_{i'}$, the zero-profit and budget balanced curves for agent $i'$ are as shown in Figure~\ref{fig:y_less_than_minus_z}. The utility for agent $i'$ if allocated and charged $(z^\ast, y^\ast)$ is
\begin{align*}
	u_{i'}(z^\ast) - y^\ast = \fixedV_{i'}\fixedP_{i'} - (1 - \fixedP_{i'}) z^\ast - y^\ast  = \fixedV_{i'}\fixedP_{i'} + \fixedP_{i'} z^\ast + (-z^\ast - y^\ast) > 0,
\end{align*}
thus in the economy where agent $\winner$ is replaced by agent $i'$, agent $i'$ has to be allocated and cannot tie with any other agent (otherwise she gets zero utility given Lemma~\ref{lem:zero_utility_when_tied}, thus has a useful deviation). However,  since the agent's utility and the mechanism's revenue add up to $\fixedV_{i'}\fixedP_{i'}$, the mechanism's revenue is upper bounded by $ - \fixedP_{i'} z^\ast + (z^\ast + y^\ast)  < (1 - (1 + y^\ast/z^\ast)) z^\ast + y^\ast = 0$, which contradicts ND. Therefore, we conclude that $y^\ast \geq -z^\ast$ holds, and this completes the proof of the lemma. 
\end{proof}

\subsubsection{Proof of Lemma~\ref{lem:lem_P1P5_characterization}}

We are now ready to prove Lemma~\ref{lem:lem_P1P5_characterization}, which characterizes the set of possible outcomes under any mechanism that satisfies (P1)-(P6).

\lemCharacterization*

\begin{proof} Part (i) is implied by IR. We have already proved part (iii) in Lemma~\ref{lem:admissible_payments}. 
For part (ii), we show that if exists agent $i' \neq \winner$ s.t. $u_{i'}(z^\ast) - y^\ast > 0$, then there is a contradiction.
To simplify notation, consider a type profile $\CDF = (\CDF_1,\dots, \CDF_n)$, where agent $\winner = 1$ is allocated with non-zero probability, and assume that agent $2$ has $u_{2}(z^\ast) - y^\ast > 0$. 
Now consider agent $1'$, whose type is identical to that of agent $2$: $\CDF_{1'} = \CDF_2$. For economy $(\CDF_{1'}, \CDF_{-1})$, given anonymity, we know that one of the following two situations must be true: case (1), neither of agent $1'$ or $2$ ever gets allocated, in which case they both get zero utilities; case (2), both agents $1'$ and $2$ are allocated with some probability less than one (given that have the same DSE bids and the mechanism is anonymous), in which case Lemma~\ref{lem:zero_utility_when_tied} implies that it is also the case that they both get zero utility. As a result, each of them has a useful deviation, which is to report as if her type was $\CDF_i$, get allocated, and get strictly positive utility. This contradicts DSE, and completes the proof of this lemma.
\end{proof}

\subsection{Proof of Lemma~\ref{lem:range_of_ut_sw}} \label{appx:proof_lem_range_of_ut_sw}

\lemUtSwMon*

\begin{proof}

For any economy, and any mechanism that satisfies (P1)-(P6), if there is a single agent that gets allocated, the agent has to be charged a two-part payment $(z^\ast, y^\ast)$ for some $z^\ast$. If more than one agents are tied, i.e. submitted the same reports, then anonymity requires that they are charged the same two-part payment $(z^\ast, y^\ast)$. Therefore, for each economy, each mechanism that satisfies (P1)-(P6) charges the allocated agent(s) a penalty $z^\ast$. We know from Lemma~\ref{lem:lem_P1P5_characterization} that when the penalty is $z^\ast$, the allocated agent(s) must reside on the frontier of the economy at penalty $z^\ast$.

Now consider an economy, and two mechanisms $\mech_1$ and $\mech_2$ both satisfying (P1)-(P6), which charge the allocated agents penalties $z_1^\ast$ and $z_2^\ast$ respectively, and assume without loss that $z_1^\ast < z_2^\ast$. Let $N_1 \triangleq \arg \max_{i \in N} u_i(z_1^\ast)$ and let $N_2 \triangleq \arg \max_{i \in N} u_i(z_2^\ast)$, i.e. the sets of agents that might be allocated in the two mechanisms. Let $\wone \in \arg\max_{i \in N_1} \ut_i(z_1^\ast)$, and $\wtwo \in \arg \min_{i \in N_2} \ut_i(z_2^\ast)$, we know that regardless of how the mechanism breaks ties, the the utilization achieved by $\mech_1$ is upper bounded by $\ut_{\wone}(z_1^\ast)$, and that the utilization achieved by $\mech_2$ is lower bounded by $\ut_{\wtwo}(z_2^\ast)$. 
Now consider two cases. If $\wone = \wtwo$, $\ut_{\wone}(z_1^\ast) \leq \ut_{\wtwo}(z_2^\ast)$ is implied by Lemma~\ref{lem:util_welfare}. If $\wone \neq \wtwo$, we know that $u_{\wone}(z_1^\ast) \geq u_{\wtwo}(z_1^\ast)$ and $u_{\wone}(z_2^\ast) \leq u_{\wtwo}(z_2^\ast)$ hold. In this case, Lemma~\ref{lem:crossing_utilities} implies that $\ut_\wone(z_1^\ast) \leq \ut_\wtwo(z_2^\ast)$ holds. This completes the proof that $\mech_2$ achieves (weakly) higher utilization than $\mech_1$. 

For the monotonicity properties of social welfare, when $z_1^\ast < z_2^\ast \leq \socialV$, let $\wone \in \arg\max_{i \in N_1} \sw_i(z_1^\ast)$, and $\wtwo \in \arg \min_{i \in N_2} \sw_i(z_2^\ast)$, then apply Lemma~\ref{lem:util_welfare} and Lemma~\ref{lem:crossing_utilities}, we know that regardless of whether $\wone = \wtwo$, $\mech_2$ achieves weakly better welfare than $\sw_\wtwo(z_2^\ast)$, which is weakly higher than $\sw_\wone(z_1^\ast)$, which is an upper bound on welfare achieved by $\mech_1$. The monotonicity of welfare can be proved similarly for the case when $\socialV \leq z_1^\ast < z_2^\ast$. This completes the proof of this lemma.
%
\end{proof}

\subsection{Optimality of CP($\socialV$)} \label{appx:proof_opt_cpm}

\subsubsection{Proof of Theorem~\ref{thm:cpm_not_dom}}

\thmNotDominated*

\begin{proof}

Assume that there is a mechanism $\mech$ satisfying (P1)-(P6) s.t. the social welfare under $\mech$ is always as good as that under CP($\socialV$) for every type profile $\CDF$ that satisfy (A1)-(A2). We  show that $\mech$ either has the same outcome as in CP($\socialV$), or achieve the same social welfare.

Assume w.l.o.g that agent $1$ is allocated under CP($\socialV$) given some type profile $\CDF$, and let $(z^\ast, y^\ast)$ be her two-part payment under CP($\socialV$). For the case when $\max_{i \neq 1} u_i(\socialV) \geq 0$, we know that $u_1(\socialV) > 0$, and that $z^\ast = \socialV$ under CP$(\socialV)$, and this outcome is already welfare-optimal. Thus $\mech$ cannot achieve higher social welfare, therefore what is left to consider is the case where $u_i(\socialV) < 0$ for all $i \neq 1$. Assume w.l.o.g. that $\zc_2 = \max_{i \neq 1} \zc_i$ and that $\zc_2 < \socialV$. The generic input assumption implies that $\zc_2 < \zc_1$.

We now show that $\mech$ must also allocate the resource to agent $1$ in order to achieve weakly higher social welfare than the CP($\socialV$) mechanism. First, $\mech$ must allocate the resource to some agent, otherwise it achieves worse welfare. The monotonicity property in Lemma~\ref{lem:range_of_ut_sw} requires that $\mech$ charges the allocated agent a penalty weakly higher than $z^\ast = \zc_2$. Since the base payment $y^\ast$ cannot be smaller than zero, the only agents that found a penalty of $\zc_2$ or higher acceptable are agents $1$ and agent $2$. Allocating to agent $2$ at penalty $\zc_2$ is not possible, since this violates part (ii) of Lemma~\ref{lem:lem_P1P5_characterization}. Therefore, agent $1$ must still be allocated under $\mech$.


\begin{figure}[t!]
\centering 
\begin{tikzpicture}[scale = 1.1][font=\small]
\draw[->] 	(-0.3,0) -- (4.2,0) node[anchor=north] {$z$};
\draw[->] 	(0,-0.5) -- (0,1.9) node[anchor=west] {$u_i(z)$};

\draw[-, name path = uA] (-0.2, 1.7) to [out=-35, in = 170] (4, -0.1);
\draw[dashed] (-0.2, 1) -- (1.5, -0.25);

\draw[dashdotted] (-0.2, 1.3) to [out=-35, in = 170] (4, -0.5);

\draw[loosely dotted](2, 0.4) -- (2, -0.1);

\draw[loosely dotted](2.1, 0.2) -- (-0.1, .2)node[anchor=east] {$y'$};

\draw (1.1, 0) node[anchor = north]{$\zc_2$};
\draw (2, 0) node[anchor = north]{$z'$};

\filldraw [black] (2, 0.2) circle (1pt);

\draw[-] (2.6, 1.6) -- (3.1, 1.6) node[anchor=west] {$u_1(z)$};
\draw[dashed] (2.6, 1.2) -- (3.1, 1.2) node[anchor=west] {$u_2(z)$};
\draw [dashdotted](2.6, 0.8) -- (3.1, 0.8) node[anchor=west] {$u_{1'}(z)$};

\end{tikzpicture}
\caption{Illustration for the proof of Theorem~\ref{thm:cpm_not_dom}.
\label{fig:proof_thm_not_dom}} 
\end{figure}
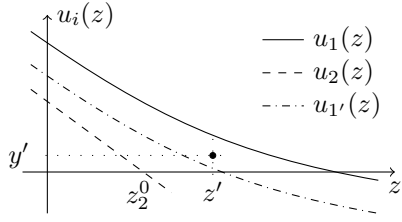

If the penalty charged by mechanism $\mech$ is also $\zc_2$, then the outcome coincides with that under CP$(\socialV)$. 
Therefore, the only case to consider is when $\mech$ charges agent $1$ a two-part payment $(z', y')$ where $z' > \zc_2$. Part (iii) of Lemma~\ref{lem:lem_P1P5_characterization} requires that $y' \geq 0$, thus we have $u_2(z') - y' < 0$. See Figure~\ref{fig:proof_thm_not_dom}. Define $\eps  \triangleq u_1(\zc_2) - (u_1(z') - y')$. We know from the monotonicity of $u_1(z)$ that $\eps \geq y' \geq 0$. If $\eps = 0$, it must be the case that $u_1(\zc_2) = u_1(z')$ since $y'$ is non-negative, in which case $u_1(z)$ is flat in this range, and as a result the right derivative of $u_1(z)$ is zero. From Lemma~\ref{lem:util_welfare}, we know that $\Pm{V_1 \geq -\zc_2} = \Pm{V_1 \geq -z'} = 1$, and $\mech$ achieves the same welfare and utilization as the CP$(\socialV)$ mechanism. 
Therefore, what is left to consider the case where $\eps > 0$. Let $u_{1'}(z) $ be
\begin{align*}
	u_{1'}(z) \triangleq \max\{ u_1(z) - (u_1(z') - y') - \eps/2, ~-z \}, ~\forall z \in \setR.
\end{align*}
See Figure~\ref{fig:proof_thm_not_dom}. For $z = \zc_2$, $u_{1'}(\zc_2) = u_1(\zc_2) - (u_1(z') - y') - \eps/2 = \eps - \eps / 2 = \eps/2 > 0$, therefore, $u_{1'}(0) \geq u_{1'}(\zc_2) > 0$. 
We know from Lemma~\ref{lem:u_to_F} that $u_{1'}(z)$ is the expected utility function of some type that satisfies (A1)-(A2). Let's call this agent $1'$, and we know that her zero-crossing satisfies $\zc_{1'} > \zc_2$. 
In the original economy, if we replace agent $1$ with agent $1'$, agent $1'$ will be selected by mechanism $\mech$ as well, since $\mech$ must allocate to the CP$(\socialV)$ winner. The two-part payment that agent $1'$ faces, which we denote as $(z'', y'')$, must satisfy $z'' \geq \zc_2$, and that $u_{1'}(z'') - y'' \geq 0$, in order to achieve welfare as high as the CP$(\socialV)$ mechanism and not violate IR. For $z'' \geq \zc_2$, we know $u_1(z'') = u_{1'}(z'') + (u_1(z') - y') + \eps/2$. Given the construction of $u_{1'}(z)$, we know that if agent $1$ reports the type of agent $1'$, she will get allocated under mechanism $\mech$, charged $(z'', y'')$, and get expected utility
\begin{align*}
	u_1(z'') - y'' = u_{1'}(z'') + (u_1(z') - y') + \eps/2 - y'' >  u_1(z') - y'.
\end{align*}
This is a useful deviation for agent $1$ under $\mech$, therefore mechanism $\mech$ cannot strictly dominate the CP$(\socialV)$ mechanism in welfare, and this completes the proof of this theorem.
\end{proof}

\subsubsection{Proof of Theorem~\ref{thm:cpm_opt}} 

Before proving the theorem, we first provide the following lemma on the ordered payment space.

\newcommand{\tldz}{\tilde{z}}
\newcommand{\tldy}{\tilde{y}}

\begin{lemma} \label{lem:payment_not_ordered} Assume that the type space includes all $(\fixedV_i, \fixedP_i)$ types. Let $(z, y)$ and $(\tldz, \tldy)$ be two sets of two-part payment in the same ordered payment space. It cannot be the case that $0 < y - \tldy < \tldz - z$. 
\end{lemma}
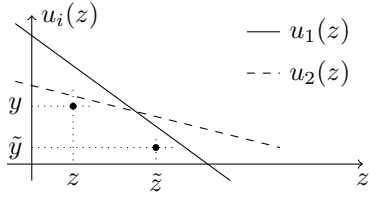
\begin{figure}[t!]
\centering 
\begin{tikzpicture}[scale = 1.1][font=\small]
\draw[->] 	(-0.3,0) -- (4,0) node[anchor=north] {$z$};
\draw[->] 	(0,-0.2) -- (0,1.8) node[anchor=west] {$u_i(z)$};

\draw[-] (-0.2, 1.7) -- (2.4, -0.2);
\draw[dashed] 	(-0.2, 1) -- (3, 0.2); 

\filldraw [black] (1.5, 0.2) circle (1pt);
\filldraw [black] (0.5, 0.7) circle (1pt);

\draw[dotted] (1.5, 0.3) -- (1.5, 0)  node[anchor = north]{$\tldz$};
\draw[dotted] (0.5, 0.9) -- (0.5, 0)  node[anchor = north]{$z$};
\draw[dotted] (1.7, 0.2) -- (0, 0.2)  node[anchor = east]{$\tldy$};
\draw[dotted] (0.7, 0.7) -- (0, 0.7)  node[anchor = east]{$y$};

\draw[-] (2.6, 1.6) -- (3, 1.6) node[anchor=west] {$u_1(z)$}; 
\draw[dashed] (2.6, 1.1) -- (3, 1.1) node[anchor=west] { $u_2(z)$};
\end{tikzpicture}

\caption{Illustration for the proof of Lemma~\ref{lem:payment_not_ordered}.
\label{fig:proof_lemma_orderedP}} 
\end{figure}

\begin{proof}

Assume towards a contradiction, that there exists an ordered payment space $\paymentSpace$, and two-part payments $(z, y),~(\tldz, \tldy) \in \paymentSpace$ s.t. $0 < y - \tldy < \tldz - z$, as illustrated in Figure~\ref{fig:proof_lemma_orderedP}. Denote $q = 1 - (y - \tldy)/(\tldz - z)$, we know that $q \in (0,1)$. Now consider two $(\fixedV_i, \fixedP_i)$ types, with $\fixedV_1, \fixedV_2 > \max\{-z, -\tldz, 0\}$, and with $\fixedP_1 = q/2$, and $\fixedP_2 = q/2 + 1/2$. We know that $0 < \fixedP_1 < q < \fixedP_2 < 1$, and that $(\fixedV_1, \fixedP_1)$ and $(\fixedV_2, \fixedP_2)$ are both valid parameters of $(\fixedV_i, \fixedP_i)$ types. Moreover:
\begin{align*}
	u_1(z) - y - (u_1(\tldz) - \tldy) = -(1-\fixedP_1)(z - \tldz)  - y + \tldy  > -(1-q)(z - \tldz)  - y + \tldy = 0, \\
	u_2(z) - y - (u_2(\tldz) - \tldy) = -(1-\fixedP_2)(z - \tldz)  - y + \tldy < -(1-q)(z - \tldz)  - y + \tldy = 0,
\end{align*}
meaning that agents $1$ and $2$ has different, strict preferences over $(z,y)$ and $(\tldz, \tldy)$. 
See Figure~\ref{fig:proof_lemma_orderedP}. This is a contradiction to the definition of an ordered payment space.
\end{proof}

We now prove the optimality result. 

\thmOptOrderedSpace*

\begin{proof} Let $\mech$ be a mechanism that satisfies (P1)-(P6), always allocates the resource, and uses an ordered payment space $\paymentSpace$. We know from Lemma~\ref{lem:lem_P1P5_characterization} and Lemma~\ref{lem:range_of_ut_sw} that given any economy, if mechanism $\mech$ achieves better social welfare than the CP$(\socialV)$ mechanism, then it must allocate to the same agent as the CP$(\socialV)$ mechanism does, and charge a higher penalty than the second highest zero-crossing. We show a conflict to the ordered payment space if this is the case.

\newcommand{\shift}{0.45}

\begin{figure}[t!]
\centering 
\begin{tikzpicture}[scale = 1.1][font=\small]
\draw[->] 	(-2,0) -- (4.2,0) node[anchor=north] {$z$};
\draw[->] 	(0,-0.5) -- (0,2.5) node[anchor=west] {$u_i(z)$};

\draw[-, name path = uA] (-1.25, 2.4) to [out=-34, in = 170] (4, 0.1);
\draw[dashed] (-1.2, 1.2) to [out=-25, in = 170] (3, -0.2);
\draw[dashdotted] (-1.7, 2.22) -- (1, -0.15);

\draw[dotted](-1.9, 2.4) -- (-1.18, 2.2-\shift) to [out=-30, in = 170] (4, 0.1-\shift);

\draw (0.8, 0) node[anchor = north]{$\zc_3$};
\draw (2, 0) node[anchor = north]{$\zc_2$};
\draw[loosely dotted](2, 0.6) -- (2, -0.1);

\draw[loosely dotted](1.3, 1) -- (1.3, -0.2);
\draw (1.3, -0.08) node[anchor = north]{$\xpt$};

\filldraw [black] (2.7, 0.2) circle (1pt);
\draw[loosely dotted](2.7, 0.6) -- (2.7, -0.1);
\draw (2.7, -0.1) node[anchor = north]{$z^\ast$};
\draw[loosely dotted](2.8, 0.2) -- (-0.1, .2)node[anchor=east] {$y^\ast$};

\draw[-] (2.6, 2.2) -- (3.1, 2.2) node[anchor=west] {$u_1(z)$};
\draw[dashed] (2.6, 1.8) -- (3.1, 1.8) node[anchor=west] {$u_2(z)$};
\draw[dashdotted] (2.6, 1.4) -- (3.1, 1.4) node[anchor=west] {$u_3(z)$};
\draw [dotted](2.6, 1) -- (3.1, 1) node[anchor=west] {$u_{1'}(z)$};

\end{tikzpicture}
\caption{Illustration for the proof of Theorem~\ref{thm:cpm_opt}.
\label{fig:proof_thm_cpm_opt}} 
\end{figure}
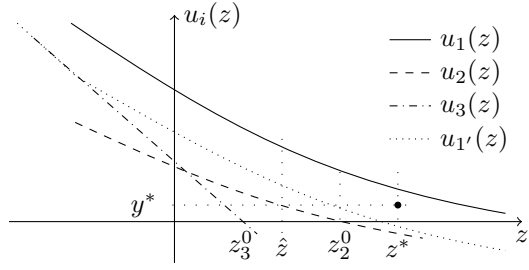

Assume w.l.o.g. that agent $1$ is allocated under mechanism $\mech$, and assume toward a contradiction, that agent $1$ is charged the two part payment $(z^\ast, y^\ast)$ with $z^\ast > \zc_2 = \max_{i \neq 1} \zc_i$, as illustrated in Figure~\ref{fig:proof_thm_cpm_opt}. If $\mech$ achieves better welfare, it must be the case that $\zc_2 < \socialV$. IR requires that $\pi \triangleq u_1(z^\ast) - y^\ast \geq 0$. 
If $u_1(\zc_2) = \pi + y^\ast = u_1(z^\ast)$, from the convexity of $u_1(z)$, we know that 
$u_1'(z+) = 0$ for all $z \geq \zc_2$, meaning that $\Pm{V_1 \geq \zc_2} = \Pm{V_1 \geq z^\ast} = 1$, i.e. in comparison to the CP$(\socialV)$ mechanism, although $\mech$ charges a higher penalty, it has no improvement on welfare or utilization. Therefore, we only need to consider the case when $u_1(\zc_2) > \pi + y^\ast$. Define $\eps \triangleq \pi + (u_1(\zc_2) - \pi - y^\ast)/2$, we know that $\eps > \pi$.

Now consider $u_{1'}(z)$ defined as $u_{1'}(z) \triangleq \max\{u_1(z) - \eps, ~ u_{ N \backslash \{1\} }  (z) \}$ for all $z \in \setR$, as illustrated in Figure~\ref{fig:proof_thm_cpm_opt}. Lemma~\ref{lem:u_to_F} implies that $u_{1'}(z)$ is the expected utility function of some agent, whose type, which we denote as $\CDF_{1'}$, also satisfies (A1)-(A2). Now consider the economy $(\CDF_{1'}, \CDF_{-1})$, and assume that the allocated agent is charged a two-part payment $(z',y')$ in this economy. We need to consider two cases, depending on whether agent $1'$ is allocated.

\emph{Case 1:} Assume agent $1'$ is allocated, in which case $u_{1'}(z') - y' \geq 0$ given IR. If the payment $(z',y')$ satisfies $u_{1'} (z') - y' > 0$, we know that it cannot be the case that $u_{1'}(z') = u_{ N \backslash \{1\} } (z')$. This is because $u_{1'}(z') = u_{ N \backslash \{1\} } (z')$ implies that $u_{ N \backslash \{1\} } (z') - y' > 0$ holds, and this violates part (ii) of Lemma~\ref{lem:lem_P1P5_characterization}. As a result, $u_{1'}(z') = u_{1} (z') - \eps$ given the construction of $u_{1'}(z)$, and this implies a useful deviation for agent $1$ in the original economy: if she pretends that her type is actually $\CDF_{1'}$, gets allocated and charged $(z', y')$, her utility will be $u_1(z') - y' = u_{1'}(z') + \eps - y' > \eps > \pi$. As a result, we must have $u_{1'} (z') - y' = 0$. We claim that it must also be the case that $u_{ N \backslash \{1\} } (z') - y' = 0$, since otherwise, it has be the case that $u_{ N \backslash \{1\} } (z') - y' < 0$, which again implies $u_{1'}(z') = u_{1} (z') - \eps$ holds, and this also results in a useful deviation for agent $1$ in the original economy.

\emph{Case 2:} If agent $1'$ is not allocated in the economy $(\CDF_{1'}, \CDF_{-1})$, the payment $(z',y')$ satisfies $u_{ N \backslash \{1\} } (z') - y' \geq 0$, since some other agent has to be allocated. Part (ii) of Lemma~\ref{lem:lem_P1P5_characterization}  and the construction of $u_{1'}(z)$ imply that $u_{1'}(z') - y' \leq 0 \Rightarrow u_{ N \backslash \{1'\} } (z') - y' \leq 0$. As a result, $u_{ N \backslash \{1'\} } (z') - y' =0$ holds, which also implies $u_{1'}(z') - y' = 0$, since $u_{1'}(z') \geq u_{ N \backslash \{1'\} } (z')$.

Therefore, for both cases, we know $u_{1'}(z')-y' = 0$ and $u_{ N \backslash \{1\} }(z') - y' = 0$ both hold, implying $u_{1'}(z') = u_{ N \backslash \{1\} }(z')$. Given part (iii) of Lemma~\ref{lem:lem_P1P5_characterization}, we know that $y' \geq 0$, therefore $z' \leq \zc_2$ since otherwise $u_{ N \backslash \{1\} }(z') - y' < 0$. 
Let $\xpt$ be the unique point where $u_{ N \backslash \{1'\}}(\xpt) = y^\ast$. Since $y^\ast \geq 0$ and $\zc_2 < \socialV$ (which implies $u_{ N \backslash \{1'\}}(\zc_2 ) < 0$), we know that $\xpt$ is well defined given Lemma~\ref{lem:exp_u} and Lemma~\ref{lem:exp_u_appx}, and that $\xpt$ is strictly smaller than $\zc_2$. The monotonicity of $u_1(z)$ guarantees that for all $z \in [\xpt, \zc_2]$, $u_1(z) \geq u_1(\zc_2)$, therefore $u_{1'}(z) \geq u_1(z) - \eps \geq u_1(\zc_2) - \pi - (u_1(\zc_2) - \pi - y^\ast  )/2 = y^\ast + (u_1(\zc_2) - \pi - y^\ast  )/2 > y^\ast$. As a result, $u_{1'}(z) > y^\ast \geq u_{ N \backslash \{1'\}}(z)$ holds for all $z \in [\xpt, \zc_2]$, thus $z' \geq \xpt$ violates $u_{1'}(z') = u_{ N \backslash \{1\} }(z')$. As a result, we know $z' < \xpt$.

We now claim that $0 < y' - y^\ast < z^\ast - z'$ holds, which violates the ordered payment space assumption, according to Lemma~\ref{lem:payment_not_ordered}. First, $u_{ N \backslash \{1\} }  (z') - y' = 0$, $u_{ N \backslash \{1\} }  (\xpt) - y^\ast = 0$ and $z' < \xpt$ implies that $y' > y^\ast$. Moreover, $u_{ N \backslash \{1\} }  (z') - y' = 0$, $u_{ N \backslash \{1\} }  (\xpt) - y^\ast = 0$ and the fact that the slope of $u_{ N \backslash \{1\} } (z)$ is lower bounded by $-1$ implies that $y' - y^\ast < \xpt - z' < \zc_2 - z' < z^\ast - z'$. Thus we get $0 < y' - y^\ast < z^\ast - z'$.
This completes the proof that no mechanism with the stated properties can charge the allocated agent a higher penalty the penalty under CP$(\socialV)$ in any economy and improve welfare. We thus conclude that the CP$(\socialV)$ mechanism is welfare-optimal profile by profile.    
\end{proof}

\subsection{Proof of Theorem~\ref{thm:csp_uniq_opt}} \label{appx:proof_opt_csp}

\thmUniqOptCSP*

\begin{proof}

With the generic input assumption, we consider w.l.o.g an economy, where agents are ordered in decreasing-order in their zero-crossings $\zc_1 > \zc_2 > \dots > \zc_n$. Agent $1$ is assigned the resource under CSP, and she is charged penalty $z^\ast = \zc_2$, and a zero base payment $y^\ast = 0$. 
A similar analysis as in the proof of Theorem~\ref{thm:cpm_not_dom} (see Appendix~\ref{appx:proof_opt_cpm}) shows that the only way to achieve a higher utilization than the outcome under CSP is to also assign the resource to agent $1$, and charge a penalty that is higher than $\zc_2$. This immediately give us part (ii) of this theorem, since for the $(\fixedV_i, \fixedP_i)$ types, a higher penalty does not improve utilization. The exact same construction as in the proof of Theorem~\ref{thm:cpm_not_dom} shows that charging a higher penalty results in a violation of one of DSE, IR or ND, unless the mechanism charges a smaller penalty than CSP does in some other economy. Therefore the CSP is not dominated in utilization by any mechanism, which is part (iii).

We now prove part (i) of this theorem. From Lemma~\ref{lem:lem_P1P5_characterization} and the ``no charge" assumption, we know that for any type profile $\CDF$, the payment two-part $(z^\ast, y^\ast)$ facing the allocated agent $\winner$ facing the allocated agent $\winner$ must satisfy $y^\ast=0$ and $z^\ast \in [\zc_2, ~\zc_1]$. The only agents that can be allocated are agents $1$ and $2$, without violating IR. Allocating to agent $2$ violates part (ii) of Lemma~\ref{lem:lem_P1P5_characterization}, therefore agent $1$ is allocated. We now argue $z^\ast = \zc_2$ must hold, and thus the CSP outcome. Assume otherwise, that there exist a mechanism $\mech$ and an economy where and $z^\ast > \zc_2$. Consider the economy in which agent 1 is replaced by agent $1'$ with zero-crossing $\zc_{1'} \in (\zc_{2}, z^\ast)$. Agent $1'$ must be allocated from the above argument and the requirement that the resource is allocated, thus the penalty that she faces $z'$ must be smaller than $z^\ast$ (otherwise IR is violated). This gives agent $1$ in the original economy a useful deviation, which is to report the type of agent $1'$, getting allocated and charged a smaller penalty. This is a contradiction, and therefore proves the uniqueness of CSP.

Part (iv) can be proved by the same arguments as in the proof of Theorem~\ref{thm:cpm_opt}: if a mechanism charges the allocated agent a two-part payment $(z^\ast, y^\ast)$ with $z^\ast > \zc_2$ and improves utilization, then we can construct an alternative economy, where the new payment $(z', y')$ collected by the mechanism cannot reside in the same ordered payment space together with $(z^\ast, y^\ast)$. 
%
%
\end{proof}

\section{Additional Simulation Results} \label{appx:additional_simulations}



For the assignment of a single resource, the value for each agent to use the resource is $V_i = w_i - O_i$, where $O_i$ is the random opportunity cost with exponential distribution parametrized by $\lambda_i$. For uniform distribution of the opportunity cost $\lambda_i^{-1} \sim \mathrm{U}[0,~10]$ and random distribution of the value $w_i \sim \mathrm{U}[0, \lambda_i^{-1}]$, we examined the setting when $\socialV = 5$ is equal to the average opportunity cost. 

For the scenario where $\socialV = 0$, meaning that the society does not care about whether the resource is utilized, the average welfare and utilization achieved by the mechanisms are as shown in Figure~\ref{fig:W0}. The generalized CP(0) mechanism coincides with the second price auction, and both achieve the first-best social welfare. The CSP mechanism has higher utilization, but achieves lower welfare, since when the society does not value utilization, incentivizing agents to use the resource when they naturally would not hurts the welfare.

\begin{figure}[t!]
\centering
\subfloat[Social Welfare]{\label{fig:welfare_W0}
	\includegraphics[scale=\figScale]{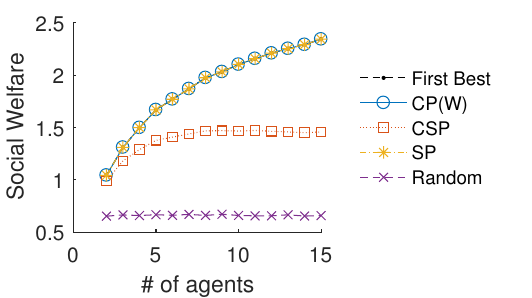}
}
%
\subfloat[Utilization]{\label{fig:utilization_W0}
 	\includegraphics[scale=\figScale]{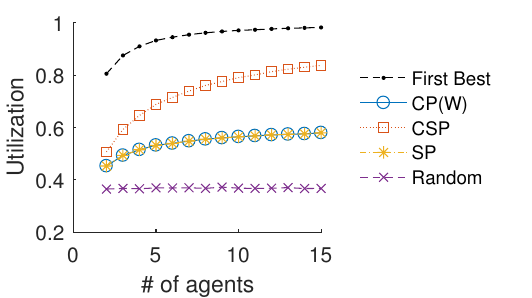}
}
\caption{Social welfare and utilization for a single resource, with small societal value $\socialV = 0$.
\label{fig:W0} 
}
\end{figure}

For the setting where $\socialV = 10$, modeling the scenario where the society have a stronger preference for the resource being utilized, 
%
the average welfare and utilization under different mechanisms for 10,000 randomly generated economies are as shown in Figure~\ref{fig:W10}.  
Figure~\ref{fig:welfare_W10} shows that when the societal value $\socialV$ is higher, the CP$(\socialV)$ mechanism achieves a much higher social welfare than the SP mechanism. The CP$(\socialV)$ and CSP social welfare coincide since when $\socialV$ is high, no agent is willing to accept $\socialV$ as a period~1 penalty thus the outcomes under the two mechanisms are the same. Their performance is competitive to the first-best benchmark under IR and ND, especially when the number of agents is large.
Figure~\ref{fig:utilization_W10} presents the average utilization under these mechanisms. CP($\socialV$) and CSP coincide, and both achieve much higher utilization than SP.

\begin{figure}[t!]
\centering
\subfloat[Social Welfare]{\label{fig:welfare_W10}
	\includegraphics[scale=\figScale]{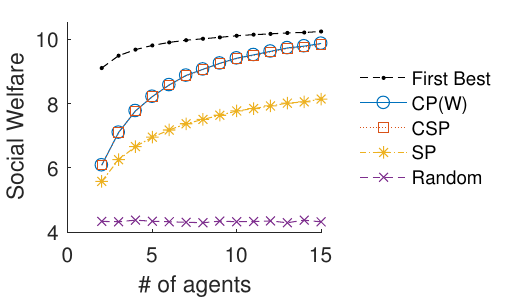}
}
%
\subfloat[Utilization]{\label{fig:utilization_W10}
 	\includegraphics[scale=\figScale]{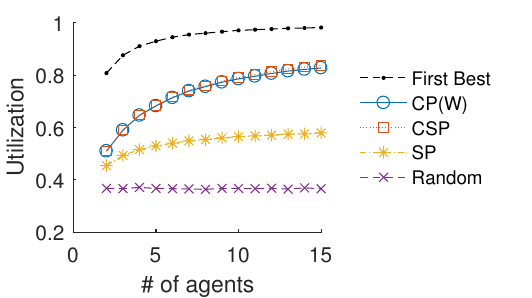}
}
\caption{Social welfare and utilization for a single resource, with large societal value $\socialV = 10$.
\label{fig:W10} 
}
\end{figure}

\end{document}